\documentclass[envcountsame]{llncs}

\newcommand{\Lang}{\mathsf{L}}

\usepackage{amssymb,amsmath,mathtools}
\usepackage{wrapfig}
\usepackage{tikz}
\usetikzlibrary{automata}

\usepackage{fixltx2e}
\usepackage{textcomp} %tilde
\usepackage{xspace}

\newcommand{\Nset}{\mathbb N}
\newcommand{\X}{{\ensuremath{\mathbf{X}}}}
\newcommand{\F}{{\ensuremath{\mathbf{F}}}}
\newcommand{\G}{{\ensuremath{\mathbf{G}}}}
\newcommand{\U}{{\ensuremath{\mathbf{U}}}}
\newcommand{\true}{{\ensuremath{\mathbf{tt}}}}
\newcommand{\false}{{\ensuremath{\mathbf{ff}}}}

\newcommand{\gsf}{{\ensuremath{\mathbb{G}}}}

\newcommand{\tran}[1]{\stackrel{#1}{\longrightarrow}}

\newcommand{\A}{\mathcal{A}}

\newcommand{\setG}{\mathcal{G}}
\newcommand{\K}{\mathcal{M}}

\renewcommand{\P}{\mathcal{P}}
\newcommand{\ka}{Mojmir\xspace}

\newcommand{\aft}{{\it af}}
\newcommand{\aftg}{{\it af}_\G}

\newcommand{\Reach}{{\it Reach}}
\newcommand{\Reachg}{{\it Reach}_\G}
\newcommand{\R}{\mathcal{R}}

\newcommand{\rank}[1]{\mathit{rk}(#1)}
\newcommand{\rankf}[1]{\mathcal{F}(#1)}

\newcommand{\ind}[2]{\mathit{ind}(#1,#2)}
\newcommand{\M}{\mathcal{T}}

\newcommand{\theoremlike}[2]{\par\medskip\penalty-250\refstepcounter{theorem}{{\bfseries\noindent#2
\ref{#1}.}}}

\newcommand{\thmhelperpre}[2]{\theoremlike{#1}{#2}}
\newcommand{\thmhelperpost}{\par\medskip}

\newenvironment{reftheorem}[1]{\thmhelperpre{#1}{Theorem}}{\thmhelperpost}
\newenvironment{reflemma}[1]{\thmhelperpre{#1}{Lemma}}{\thmhelperpost}

\title{From LTL to Deterministic Automata:\\ A Safraless Compositional Approach}

\author{Javier Esparza and Jan K\v ret\'insk\'y%\inst{1,2,3}
\thanks{This research was funded in part by the European Research Council (ERC) under grant agreement 267989
(QUAREM) and by the Austrian Science Fund (FWF) project S11402-N23 (RiSE). The author is on leave from Faculty of Informatics, Masaryk University, Czech Republic, and partially supported by the Czech Science Foundation, grant No.~P202/12/G061.}\thanks{
In the original paper published in CAV'14, the negative conjuncts in Theorem~\ref{thm:master} were missing.
We thank Salomon Sickert for his help in correcting the theorem and its proof. 
}}

\institute{Institut f\"ur Informatik, Technische Universit\"at M\"unchen, Germany\\
IST Austria
}

\newcommand{\myspace}{\vspace*{-1em}}
\newcommand{\myspaceb}{\vspace*{-0.5em}}

\begin{document}

\tikzset{
    state/.style={
		rectangle,
            rounded corners,
            draw=black,
            minimum height=2em,
            minimum width=2em,
            inner sep=4pt,
            text centered,
            }
}

\maketitle 

\pagestyle{plain}

\begin{abstract}
We present a new algorithm to construct a (generalized) deterministic Rabin automaton for
an LTL formula $\varphi$. The automaton is the product of a master automaton and
an array of slave automata, one for each $\G$-subformula of $\varphi$. The slave
automaton for $\G\psi$ is in charge of recognizing whether $\F\G\psi$ holds.
As opposed to standard determinization procedures, the states of all our automata have a clear logical structure, which allows for
various optimizations. Our construction subsumes former algorithms for fragments of LTL. Experimental results show improvement in the sizes of the resulting automata compared to existing methods.
\end{abstract}
%\myspace\myspace

\section{Introduction}

Linear temporal logic (LTL) is the most popular specification language for linear-time properties. 
In the automata-theoretic approach to LTL verification, formulae are translated into
$\omega$-automata, and the product of these automata with the system is analyzed. Therefore, 
generating small $\omega$-automata is crucial for the efficiency of the approach. 

In quantitative probabilistic verification, LTL formulae need to be translated into \emph{deterministic} 
$\omega$-automata \cite{DBLP:books/daglib/0020348,cav13}. Until recently, this required to proceed in two steps: first 
translate the formula into a non-deterministic B\"uchi automaton (NBA), and then 
apply Safra's construction  \cite{DBLP:conf/focs/Safra88}, or improvements on it  
\cite{DBLP:conf/lics/Piterman06,DBLP:conf/fossacs/Schewe09} to transform the NBA into a deterministic automaton 
(usually a Rabin automaton, or DRA). This is also the approach adopted in \texttt{PRISM}~\cite{prism}, 
a leading probabilistic model checker, which reimplements the optimized Safra's construction 
of \texttt{ltl2dstar}~\cite{ltl2dstar}. 

In \cite{cav12} we presented an algorithm that {\em directly} constructs a generalized DRA (GDRA) for
the fragment of LTL containing only the temporal operators $\F$ and $\G$. The GDRA can be either (1) 
degeneralized into a standard DRA, or (2) used directly in the probabilistic verification process \cite{cav13}. 
In both cases we get much smaller automata for many formulae. For instance, the standard approach
translates a conjunction of three fairness constraints into an automaton 
with over a million states, while the algorithm of \cite{cav12} yields 
a GDRA with one single state (when acceptance is defined on transitions), and a DRA with 462 states. 
%For four fairness constraints, the transition-based DGRA is still of size 1, the DRA has 20736 states,
%and the standard approach cannot compute an automaton at all \cite{cav13,lpar13,safras}. 
%the state-sed GDRA has 256 states and the respective degeneralized DRA is 20736 times bigger, while the 
%traditional methods cannot compute the automaton at all.
In \cite{atva12,atva13} our approach was extended to larger fragments of LTL 
containing the $\X$ operator and restricted appearances of $\U$, but
a general algorithm remained elusive.

In this paper we present a novel approach able to handle full LTL,
and even the alternation-free linear-time $\mu$-calculus. The approach 
is {\em compositional}: the automaton is obtained as a parallel composition
of automata for different parts of the formula,  running in lockstep\footnote{We could 
also speak of a product of automata, but the operational view behind the term
parallel composition helps to convey the intuition.}.
%This allows for local optimizations to reduce the size of each component 
%before constructing theirproduct. 
More specifically, the automaton is the parallel composition of a \emph{master automaton}
and an array of \emph{slave automata}, one for each $\G$-subformula of the original formula, 
say $\varphi$. Intuitively, the master monitors the formula that remains to be fulfilled 
(for example, if $\varphi = (\neg a \wedge \X a) \vee \X\X\G a$, then the remaining formula
after $\emptyset \{a\}$ is $\true$, and after $\{a\}$ it is $\X\G a$),
and takes care of checking safety and reachability properties. The slave for
a subformula $\G\psi$ of $\varphi$ checks whether $\G\psi$ {\em eventually} holds, 
i.e., whether $\F\G\psi$ holds. It also monitors the formula that remains 
to be fulfilled, but only partially: more precisely, it does not monitor any
$\G$-subformula of $\psi$, as other slaves are responsible for them.
For instance, if $\psi = a \wedge \G b \wedge \G c $, then the slave for $\G\psi$ only checks that eventually 
$a$ always holds, and ``delegates'' checking $\F\G b $ and $\F\G c$ to other slaves. Further, 
and crucially, the slave may provide the information that not only $\F\G\psi$, but a stronger
formula holds; the master needs this to decide that, for instance, not only $\F\G\varphi$ 
but even $\X\G\varphi$ holds. 

The acceptance condition of the parallel composition of master and slaves is 
a disjunction over all possible subsets of $\G$-subformulas, and all possible
stronger formulas the slaves can check. The parallel composition accepts 
a word with the disjunct corresponding to the subset of formulas which hold in it.

The paper is organized incrementally. In Section \ref{sec:rabinization} we show
how to construct a DRA for a formula $\F\G\varphi$, where $\varphi$ has no occurrence
of $\G$. This gives the DRA for a bottom-level slave.
Section \ref{sec:gen-FG} constructs a DRA for an arbitrary formula $\F\G\varphi$,
which gives the DRA for a general slave, in charge of a formula 
that possible has $\G$-subformulas. Finally, Section \ref{sec:gen} constructs a DRA for 
arbitrary formulas by introducing the master and its parallel composition with
the slaves. Full proofs can be found in Appendix. %\cite{techrep}.

\myspace 

\paragraph{Related work} 
%A construction for a fragment extending the $\F,\G$ fragment 
%mentioned above appeared in \cite{DBLP:conf/atva/BabiakBKS13}. 
There are many constructions translating 
LTL to NBA, e.g., \cite{DBLP:conf/fm/Couvreur99,DBLP:conf/cav/DanieleGV99,DBLP:conf/concur/EtessamiH00,DBLP:conf/cav/SomenziB00,ltl2ba,DBLP:conf/forte/GiannakopoulouL02,DBLP:conf/wia/Fritz03,ltl3ba,spot-atva13}. The one recommended by ltl2dstar and used in PRISM is LTL2BA \cite{ltl2ba}.
Safra's construction with optimizations described in \cite{DBLP:conf/wia/KleinB07} has been implemented in 
ltl2dstar \cite{ltl2dstar}, and reimplemenetd in PRISM \cite{prism}. 
A comparison of LTL translators into deterministic $\omega$-automata can be found in 
\cite{DBLP:conf/lpar/BlahoudekKS13}.
%\paragraph{Outline of the paper}

\section{Linear Temporal Logic}
\label{sec:prelim}

In this paper, $\Nset$ denotes the set of natural numbers including zero.
``For almost every $i \in \mathbb{N}$'' means for all but finitely many $i \in \mathbb{N}$.

\medskip

This section recalls the notion of linear
temporal logic (LTL). % \cite{DBLP:conf/focs/Pnueli77}. 
We consider the negation normal form and we have the future operator explicitly in the syntax:

\begin{definition}[LTL Syntax]\label{def:ltl-syn}
The formulae of the linear temporal logic (LTL)
are given by the following syntax:
\begin{align*}
\varphi::= & \true \mid \false \mid a\mid \neg a\mid \varphi\wedge\varphi \mid \varphi\vee\varphi \mid \X\varphi \mid \F\varphi \mid \G\varphi \mid \varphi\U\varphi
\end{align*} 
over a finite fixed set $Ap$ of atomic propositions.
\end{definition}

\begin{definition}[Words and LTL Semantics]
Let $w\in (2^{Ap})^\omega$ be a word. The $i$th letter of $w$ is denoted $w[i]$, i.e.~$w=w[0]w[1]\cdots$. 
We write $w_{ij}$ for the finite word $w[i] w[i+1] \cdots w[j]$, and $w_{i\infty}$
or just $w_i$ for the suffix $w[i] w[i+1] \cdots $.

The semantics of a formula on a word $w$ is defined inductively as follows: 
$$\begin{array}[t]{lclclcl}
  w \models \true                 \\
  w \not\models \false            \\
  w \models a & \iff & a \in w[0] \\
  w \models \neg a & \iff & a \notin w[0] \\
  w \models \varphi \wedge \psi & \iff & w \models \varphi \text{ and } w \models \psi\\
  w \models \varphi \vee \psi & \iff & w \models \varphi \text{ or } w \models \psi
\end{array}
\quad
\begin{array}[t]{lcl}
w \models \X \varphi & \iff & w_1 \models \varphi \\
w \models \F \varphi & \iff & \exists \, k\in\Nset: w_k \models \varphi \\
w \models \G \varphi & \iff & \forall \, k\in\Nset: w_k \models \varphi\\
w \models \varphi\U \psi & \iff &
\begin{array}[t]{l}
\exists \, k\in\Nset: w_k \models \psi \text{ and } \\
\forall\, 0\leq j < k: w_j\models \varphi
\end{array}
\end{array}$$
\end{definition}

\begin{definition}[Propositional implication]
Given two formulae $\varphi$ and $\psi$, we say that $\varphi$ 
{\em propositionally implies} $\psi$, denoted by  $\varphi \models_p \psi$, 
if we can prove $\varphi \models \psi$ using only the axioms of propositional logic.
We say that $\varphi$ and $\psi$ are {\em propositionally equivalent}, denoted
by $\varphi \equiv_p \psi$, if $\varphi$ and $\psi$ propositionally imply each other.
\end{definition}
%\todo{Proving $\models_p$ is then NP-complete. Do we want/need this?}
%\todo{The "cleaning" is now gone because we consider  $\models_p$ in full generality anyway}

\begin{remark}
We consider formulae up to propositional equivalence, i.e., $\varphi = \psi$ means 
that $\varphi$ and $\psi$ are propositionally equivalent. Sometimes (when there is risk of confusion)
we explicitly write $\equiv_p$ instead of $=$.
\end{remark}

\subsection{The formula $\aft(\varphi, w)$}

Given a formula $\varphi$ and a finite word $w$, we define a formula
$\aft(\varphi,w)$, read ``$\varphi$ after $w$''. Intuitively, it
is the formula that any infinite continuation $w'$ must satisfy for $ww'$ to satisfy
$\varphi$.
 
\begin{definition}
Let $\varphi$ be a formula and $\nu \in 2^{Ap}$. We define the formula
$\aft(\varphi,\nu)$ as follows: 
$$\begin{array}[t]{lcl}
\aft(\true,\nu)& = & \true \\
\aft(\false,\nu)& = & \false \\
\aft(a,\nu)& = & \left\{
\begin{array}{ll} 
\true & \mbox{if $a \in \nu$} \\ 
\false & \mbox{if $a \notin \nu$}\end{array}\right. \\
\aft(\neg a,\nu)&=& \neg \aft(a,\nu) \\
\aft(\varphi\wedge\psi,\nu)&=& \aft(\varphi,\nu)\wedge\aft(\psi,\nu) \\
\aft(\varphi\vee\psi,\nu)&=& \aft(\varphi,\nu)\vee\aft(\psi,\nu) \\
\end{array}
\quad
\begin{array}[t]{lcl}
\aft(\X\varphi,\nu)& = &  \varphi\\
\aft(\G\varphi,\nu)&= & \aft(\varphi,\nu)\wedge \G\varphi\\ 
\aft(\F\varphi,\nu)&= & \aft(\varphi,\nu)\vee \F\varphi\\
\aft(\varphi\U\psi,\nu)&=& \aft(\psi,\nu)\vee(\aft(\varphi,\nu)\wedge \varphi\U\psi)
\end{array}$$
We extend the definition to finite words as follows:
$\aft(\varphi, \epsilon) = \varphi$ and $\aft(\varphi, \nu w) = \aft(\aft(\varphi,\nu),w)$.
Finally, we define $\Reach(\varphi) = \{ \aft(\varphi,w) \mid w \in (2^{Ap})^* \}$.
\end{definition}

\begin{example}
\label{ex:aft}
Let ${\it Ap} = \{a,b,c\}$, and consider the formula $\varphi = a \vee (b \; \U \; c) $. For example, we have $\aft(\varphi, \{a\}) = \true$ $\aft(\varphi, \{b\}) = (b \; \U \; c)$, $\aft(\varphi, \{c\}) = \true$, and $\aft(\varphi, \emptyset) = \false$. $\Reach(\varphi) = \{ \varphi, \alpha \wedge \varphi, \beta \vee \varphi, \true, \false\}$,
and $\Reach(\varphi)=\{a \vee (b \; \U \; c), (b \; \U \; c), \true, \false \}$.
\end{example}

%\begin{example}
%\label{ex:aft}
%Let ${\it Ap} = \{a\}$, and, for the sake of readability, let
%$\alpha = \{a\}$ and $\beta = \emptyset$ be the two letters of $2^{\it Ap}$.
%Consider the formula $\varphi = (\X \alpha) \; \U \; (\beta \wedge \X\beta)$. We have
%$\aft(\varphi, \alpha) = \alpha \wedge \varphi$, $\aft(\varphi, \beta) = \beta \vee (\alpha \wedge \varphi) \equiv_p \beta \vee \varphi$, and 
%$\Reach(\varphi) = \{ \varphi, \alpha \wedge \varphi, \beta \vee \varphi, \true, \false\}$.
%\end{example}

\begin{lemma}
\label{lem:fundaft}
Let $\varphi$ be a formula, and let $ww' \in (2^{Ap})^\omega$ be an arbitrary word. 
Then $ww' \models \varphi$ if{}f $w' \models \aft(\varphi,w)$.
\end{lemma}
\begin{proof}
Straightforward induction on the length of $w$.
\qed\end{proof}

\section{DRAs for simple $\F\G$-formulae}
\label{sec:rabinization}
%\myspaceb

We start with formulae $\F\G\varphi$ where $\varphi$ is $\G$-free, i.e., contains no occurrence of $\G$. 
The main building block of our paper is a procedure to 
construct a DRA recognizing $\Lang(\F\G\varphi)$. (Notice that
even the formula $\F\G a$ has no deterministic B\"uchi automaton.)
We proceed in two steps. First we introduce {\ka automata} and construct a 
{\ka automaton} that clearly recognizes $\Lang(\F\G\varphi)$.
We then show how to transform {\ka automata} into equivalent DRAs.

A \ka automaton\footnote{Named in honour of Mojm\'ir K\v ret\'insk\'y, father of one of the authors} is a deterministic automaton that,
at each step, puts a fresh token in the initial state, and moves all 
older tokens according to the transition function. 
The automaton accepts if all but finitely many tokens eventually reach an accepting state. 

\begin{definition}
A {\ka automaton} $\K$ over an alphabet $\Sigma$ is a tuple 
$(Q,i,\delta, F)$, where $Q$ is a set of states, 
$i \in Q$ is the initial state, 
$\delta \colon Q \times \Sigma \rightarrow Q$ is a transition function, and 
$F \subseteq Q$ is a set of accepting states satisfying $\delta(F,\Sigma)\subseteq F$,
i.e., states reachable from final states are also final.
%\todo{polish? Yes, I've added an explanatory sentence}

The {\em run} of $\K$ over a word $w[0]w[1]\cdots \in (2^{Ap})^\omega$
is the infinite sequence  $(q^0_0)%\tran{a_0}
(q^1_0,q^1_1)%\tran{\nu_1}
(q^2_0,q^2_1,q^2_2)\cdots$ such that 

$$q^{\mathit{step}}_{\mathit{token}}=
\begin{cases}
i&\text{if } \mathit{token}=\mathit{step},\\
\delta(q^{\mathit{step}-1}_{\mathit{token}},w[{\mathit{step}-1}])&\text{ if } \mathit{token}<\mathit{step} 
%\text{ and } q^{\mathit{step}-1}_{\mathit{token}} \in Q, \\
%\bot & \text{ if } \mathit{token}<\mathit{step}, \text{ and } q^{\mathit{step}-1}_{\mathit{token}} \in S \cup \{ \bot \}  
\end{cases}$$
\noindent A run is accepting if for almost every $\mathit{token} \in \mathbb{N}$
there exists $\mathit{step}\geq \mathit{token}$ such that $q^{\mathit{step}}_{\mathit{token}}\in F$.
\end{definition}

Notice that if two tokens reach the same state at the same 
time point, then from this moment on they ``travel together''.

The \ka automaton for a formula $\varphi$ has formulae as states.
The automaton is constructed so that, when running on a word $w$, 
the $i$-th token ``tracks'' the formula that must hold for $w_i$ to satisfy 
$\varphi$. That is, after $j$ steps the $i$-th token is on the 
formula $\aft(\varphi, w_{ij})$. There is only one accepting state here, namely the one
propositionally equivalent to $\true$. Therefore, if the $i$-th token reaches
an accepting state, then $w_i$ satisfies $\varphi$.

\begin{definition}
\label{def:slave1}
Let $\varphi$ be a $\G$-free formula. 
%Let $S$ be the formulae of $\Reach(\varphi)$ propositionally equivalent to $\true$ or $\false$. 
The \ka automaton for $\varphi$  is 
$\K(\varphi)=(\Reach(\varphi),\varphi,\aft, \{\true\})$.
\end{definition}
%\noindent A state of a Mojmir automaton is called a {\em sink} if it is not the initial state
%and all its outgoing transitions are self-loops.
\begin{example}
Figure \ref{fig:ka-dra} on the left shows the \ka automaton for the formula 
$\varphi = a \vee (b \; \U \; c)$. The notation for transitions is standard: $q_1 \tran{a+\bar{a}c} q_3$ means that there is a transitions from $q_1$ to $q_3$ for each
subset of $2^{{\it Ap}}$ that contains $a$, or does not contain $a$ and contains $c$.
\end{example}

\begin{figure}
\myspace\myspace\myspace
\begin{center}
\begin{minipage}{4cm}
\begin{tikzpicture}[x=3cm,y=1.5cm,font=\footnotesize,initial text=,outer sep=1mm]
\tikzstyle{acc}=[double]
\node[state,initial] (a) at (0,0) {$q_1: a \vee (b \, \U \, c)$};
\node[state] (b) at (0,-1) {$q_2: b \, \U \, c$};

\node[state,acc] (c) at (-0.6,-2) {$q_3:\true$};
\node[state] (d) at (0.6,-2) {$q_4:\false$};

\path[->] (a) edge node[right]{$\bar{a}b\bar{c}$} (b)
              edge [bend right= 30] node[left]{$a+\bar{a}c$} (c)
              edge [bend left=30] node[right]{$\bar{a}\bar{b}\bar{c}$} (d)
          (b) edge[loop left, max distance=6mm,in=16,out=349,looseness=15] node[below]{$\; b\bar{c}$} (b)
              edge node[right] {$c$} (c)
              edge node[left] {$\bar{b}\bar{c}$} (d)
	  (c) edge[loop left, max distance=6mm,in=168,out=192,looseness=15] node[left]{{\it true}} (c)
          (d) edge[loop left, max distance=6mm,in=12,out=353,looseness=15] node[right]{{\it true}} (d)
;
;
\end{tikzpicture}
\end{minipage}
\qquad\qquad\qquad\qquad\qquad
\begin{minipage}{4.0cm}
\begin{tikzpicture}[x=3cm,y=1.5cm,font=\footnotesize,initial text=,outer sep=1mm]
\tikzstyle{acc}=[double]
\node at (0,0.4){};
\node[state,initial] (1) at (0,0) {$(\mathbf{1},\bot)$};
\node[state] (2) at (0,-1.5) {$(\mathbf{2},\mathbf{1})$};

\path[->]  (1) edge[loop left, max distance=6mm,in=16,out=349,looseness=15] node[right]{$\begin{array}{l}t_1\colon a+\bar{a}c\\
t_2\colon \bar{a}\bar{b}\bar{c}\end{array}$} (1)
          (1) edge [bend right= 15] node[left]{$t_3 \colon \bar{a}b\bar{c}$} (2)
(2) edge [bend right= 15] node[right]{$\begin{array}{l}t_6\colon c\\t_7\colon a\bar{b}\bar{c}\\t_8 \colon \bar{a}\bar{b}\bar{c}\end{array}$} (1)
(2) edge[loop left, max distance=6mm,in=16,out=349,looseness=15] node[right]{$\begin{array}{l}t_4 \colon ab\bar{c}\\ t_5 \colon \bar{a}b\bar{c}\end{array}$} (2)
;
\end{tikzpicture}
\end{minipage}
\end{center}
\myspace\myspace
\caption{A \ka automaton for $a \vee (b \; \U \; c)$ and its corresponding DRA.}
\label{fig:ka-dra}
\myspace
\end{figure}
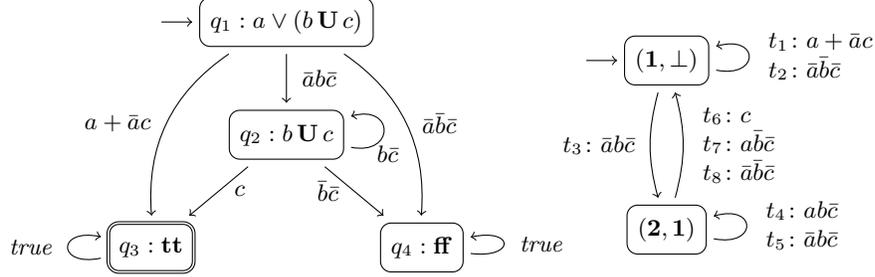

% We prove $\Lang(\K(\varphi)) = \Lang(\F\G\varphi)$ for every $\G$-free formula $\varphi$.
% We need the following lemma.
%We prove that $\K(\varphi)$ recognizes $\Lang(\F\G\varphi)$.
Since $\K(\varphi)$ accepts if{}f almost every token eventually 
reaches an accepting state, $\K(\varphi)$ accepts a word $w$ if{}f 
$w \models \F\G\varphi$. 

\begin{lemma}
\label{lem:aux}
Let $\varphi$ be a $\G$-free formula and let $w$ be a word. Then  $w \models \varphi$
if{}f $\aft(\varphi, w_{0i}) = \true$  for some $i \in \mathbb{N}$. 
\end{lemma}

\begin{theorem}
\label{thm:kretcorrect}
Let $\varphi$ be a $\G$-free formula. Then $\Lang(\K(\varphi)) = \Lang(\F\G\varphi)$.
\end{theorem}

\subsection{From \ka automata to DRAs}
\label{subsec:dekret}

Given a \ka automaton $\K=(Q,i,\delta, F)$ we construct an equivalent DRA. 
We illustrate all steps on the \ka automaton on the left of Figure \ref{fig:ka-dra}. It is convenient  
to use shorthands $q_a$ to $q_e$ for state names %rename the states 
as shown in the figure. % on the left of Figure \ref{fig:dra}. 

We label tokens with their dates of birth
%\todo{Birthday is not that good, they come once a year ...} 
(token $i$ is the token born at
 ``day'' $i$). Initially there is only one token, token $0$, placed on the initial state $i$. 
If, say, $\delta(i,\nu)=q$, then after $\K$ reads $\nu$ token $0$ moves to $q$, and token $1$ 
appears on $i$. 

%\todo{Added definition of sink\\Jan: ``sink'' occurs already in Ex.10}
A state of a Mojmir automaton is a {\em sink} if it is not the initial state
and all its outgoing transitions are self-loops. For instance, $q_3$ and $q_4$ 
are the sinks of the automaton on the left of Figure \ref{fig:ka-dra}. 
We define a {\em configuration} of $\K$ as a mapping $C \colon Q \setminus S 
\rightarrow 2^\mathbb{N}$, where $S$ is the set of sinks and $C(q)$ is the set 
of (dates of birth of the) tokens that are currently at state $q$. 
Notice that we do not keep track of tokens in sinks.

We extend the transition function to configurations: $\delta(C)$ is the configuration obtained 
by moving all tokens of $C$ according to $\delta$.
Let us represent a configuration
$C$ of our example by the vector $(C(q_1), C(q_2))$. 
For instance, we have $\delta((\{1,2\}, \{0\}),\bar{a}b\bar{c}))=(\{3\},\{0,1,2\})$.
We represent a run as an infinite sequence of configurations starting
at $(\{0\}, \emptyset)$.  The run
$(q_1)\tran{abc}(q_3,q_1)\tran{\bar{a}b\bar{c}}(q_3,q_2,q_1)\tran{\bar{a}b\bar{c}}(q_3,q_2,q_2,q_1) \cdots$ is represented by 
$(0,\emptyset) \tran{abc}(1,\emptyset)
\tran{\bar{a}b\bar{c}}(2,1)\tran{\bar{a}b\bar{c}}(3,\{1,2\})\cdots$
where for readability we identify the singleton $\{n\}$ and the number $n$. 
% 
% A token
% can {\em leave} the run by moving to a sink (this is the case of tokens 0 and 1 in our example),
% or {\em stay} in the run forever. If the token leaves by moving to an accepting sink, we say that
% it {\em succeeds}; if it leaves by moving to a non-accepting sink, we say that
% it {\em fails}.

We now define a finite abstraction of configurations.
A {\em ranking} of $Q$ is a partial function 
$r \colon Q \rightarrow \{\mathbf{1}, \ldots, \mathbf{|Q|} \}$ that assigns 
to some states $q$ a {\em rank} and satisfies:
(1) the initial state is ranked (i.e., $r(i)$ is defined) and all sinks are unranked; 
(2) distinct ranked states have  distinct ranks; and (3) if 
some state has rank $\mathbf{j}$,
then some state has rank $\mathbf{k}$ for every $\mathbf{1} \leq \mathbf{k} \leq \mathbf{j}$. 
For $\mathbf{i}<\mathbf{j}$, we say that $\mathbf{i}$ is {\em older than} $\mathbf{j}$. 
The {\em abstraction} of a configuration $C$ 
is the ranking $\alpha[C]$ defined as follows for every non-sink $q$. 
If $C(q)=\emptyset$, then $q$ is unranked. If $C(q) \neq \emptyset$, then 
let $x_q = \min\{C(q)\}$ be the oldest token in $C(q)$.
We call $x_q$ the {\em senior token} of state $q$, and  
$\{x_q \in \mathbb{N} \mid q \in Q\}$ the set of {\em senior tokens}. We 
define $\alpha[C](q)$ as the {\em seniority rank} of $x_q$: if $x_q$ is the oldest 
senior token, then $\alpha[C](q)=1$; if it is the second oldest, then $\alpha[C](q)=2$, and so on. For instance, the senior tokens of $(2,\{0,1\},\emptyset)$ are $2$ and $0$, andso $\alpha(2, \{0,1\}, \emptyset)= (\mathbf{2},\mathbf{1},\bot)$ (recall that sinks
are unranked). Notice that there are only finitely many rankings, and so 
only finitely many abstract configurations. 

The transition function $\delta$ 
can be lifted to a transition function $\delta'$ on abstract configurations by defining 
$\delta'(\alpha[C],\nu)=\alpha[\delta(C,\nu)]$. 
It is easy to see that 
$\delta'(\alpha[C],\nu)$ can be computed directly from $\alpha[C]$ (even if 
$C$ is not known). We describe how, and at the same time illustrate by computing
$\delta'( (\mathbf{2},\mathbf{1}), \bar{a}b\bar{c})$ for our running example.
\begin{itemize} 
\item[(i)] Move the senior tokens according to $\delta$. (Tokens with ranks  
$\mathbf{1}$ and $\mathbf{2}$ move to $q_2$.) 
\item[(ii)] If a state holds more than one token, keep only the most senior token.  
(Only the token with rank $\mathbf{1}$ survives.) 
\item[(iii)] Recompute the seniority ranks of the remaining tokens. (In this case unnecessary; if,
for instance, the token of rank $\mathbf{3}$ survives and the token of rank $\mathbf{2}$ does not, then 
the token of rank $\mathbf{3}$ gets its rank upgraded to $\mathbf{2}$.) 
\item[(iv)] If there is no token on the initial state, add one with the next lowest seniority rank. 
(Add a token to $q_1$ of rank $\mathbf{2}$.)  
\end{itemize} 

\begin{example}
Figure \ref{fig:ka-dra} shows on the right the transition system generated by the function $\delta'$ 
starting at the abstract configuration $(\mathbf{1},\bot)$. 
\end{example}

It is useful to think of tokens as companies that can buy other companies: at step (2),  
the senior company buys all junior companies; they all get the rank of the senior  
company, and from this moment on travel around the automaton together with the senior company.  
So, at every moment in time, every token in a non-sink state has a rank 
(the rank of its senior token). The rank of a token can age 
as it moves along the run, for two different reasons: its senior token can be bought by another  
senior token of an older rank, or all tokens of an older rank reach a sink. However, ranks can never get younger.  

Further, observe that in any run, the tokens that never reach any sink eventually get the oldest ranks, 
i.e., ranks $\mathbf{1}$ to $\mathbf{i-1}$ for some $i \geq 1$. We call these tokens \emph{squatters}.  
Each squatter either enters the set of accepting states (and stays there by assumption on \ka automata) 
or never visits any accepting state. Now, consider a run in which almost every token succeeds. 
Squatters that never visit accepting states eventually stop buying other tokens, because 
otherwise infinitely many tokens would travel with them, and thus infinitely 
many tokens would never reach final states. So the run satisfies these conditions: 
\begin{itemize}
\item[(1)] Only finitely many tokens reach a non-accepting sink (``fail''). 
\item[(2)] There is a rank $\mathbf{i}$ such that
\begin{itemize}
\item[(2.1)] tokens of rank older than $\mathbf{i}$ buy other tokens in non-accepting states 
only finitely often, and
\item[(2.2)] infinitely many tokens of rank $\mathbf{i}$ reach an accepting state (``succeed'').
\end{itemize}
\end{itemize}
\noindent Conversely, we prove that if infinitely many tokens never succeed, then (1) or (2) does not hold. 
If infinitely many tokens fail, then (1) does not hold. 
If only finitely many tokens fail, but infinitely many tokens squat in non-accepting non-sinks, then (2) 
does not hold. Indeed, since the number of states is finite, infinitely many squatters get bought in non-accepting 
states and, since ranks can only improve, their ranks 
eventually stabilize. Let $\mathbf{j-1}$ be the youngest rank such that infinitely many tokens stabilize with that rank.
Then the squatters are exactly the tokens of ranks $\mathbf{1},\ldots , \mathbf{j-1}$,
and infinitely many tokens of rank $\mathbf{j}$ reach (accepting) sinks. 
But then (2.2) is violated for every $\mathbf{i}< \mathbf{j}$, 
and (2.1) is violated for every $\mathbf{i} \geq  \mathbf{j}$ as, by the pigeonhole principle,
%\todo{more common than Dirchlet's principle}, 
there is a squatter 
(with rank older than $\mathbf{j}$) residing in non-accepting states and buying infintely many tokens.

\smallskip

So the runs in which almost every token succeeds are exactly those satisfying (1) and (2). 
We define a Rabin automaton having rankings as states, and accepting exactly these runs. We use
a Rabin condition with pairs of sets of transitions, instead of 
states.\footnote{It is straightforward to give an equivalent automaton with a condition on states, but transitions are better for us.} 
Let ${\it fail}$ be the set of transitions that move a token into a non-accepting sink.
Further, for every rank $\mathbf{j}$ let ${\it succeed}(\mathbf{j})$ 
be the set of transitions that move a token
of rank $\mathbf{j}$ into an accepting state, and ${\it buy}(\mathbf{j})$ the set of transitions
that move %a token of rank higher than or equal to $\mathbf{i}$ and 
a token of rank older than $\mathbf{j}$ and another token
into the same non-accepting state, causing one of the two to buy the other.

\begin{definition}
\label{def:rabinization}
%\todo{small changes in the definition: added condition to fail, ranks consistently in boldface, etc.}
Let $\K=(Q,i,\delta, F)$ be a \ka automaton with a set $S$ of sinks. The deterministic Rabin automaton
$\R(\K) = (Q_\R, i_\R,\delta_\R, \bigvee_{i=1}^{|Q|} P_i)$ is defined as follows:
\begin{itemize}
\item $Q_\R$ is the set of rankings $r \colon Q \rightarrow \{1, \ldots, |Q|\}$;
\item $i_\R$ is the ranking defined only at the initial state $i$ (and so $i_\R(i)=\mathbf{1}$);
\item $\delta_\R(r,\nu) = \alpha[\delta(r,\nu)]$ for every ranking
$r$ and letter $\nu$;
\item $P_j= (\mathit{fail} \cup \mathit{buy}(\mathbf{j}), \mathit{succeed}(\mathbf{j}))$, where
\begin{align*}
\mathit{fail} & = \{(r, \nu, s)\in \delta_\R \mid \exists q\in Q: r(q) \in \mathbb N \, \wedge \, \delta(q,\nu)\in S\setminus F\} \\
\mathit{succeed}(\mathbf{j})&=\{(r,\nu,s)\in \delta_\R \mid \exists q\in Q: r(q)=\mathbf{j} \, \wedge \, \delta(q,\nu)\in F \} \\
\mathit{buy}(\mathbf{j})&=\{(r,\nu,s)\in \delta_\R \mid  
\begin{array}[t]{l}
\exists q,q'\in Q:r(q)<\mathbf{j} \, \wedge  \, r(q')\in\mathbb N \\
\wedge \, \big(\delta(q,\nu)=\delta(q',\nu)\notin F\vee \delta(q,\nu)=i\notin F\big)\} \end{array}
\end{align*}
\end{itemize}
We say that a word $w \in \Lang(\R(\K))$ is accepted {\em at rank $\mathbf{j}$} if $P_j$ is
the accepting pair in the run of $\R(\K)$ on $w$ with smallest index. The rank at which
$w$ is accepted is denoted by $\rank{w}$.
\end{definition}
%\todo{cornercase: one-state automaton - sink with no transitions. Javier: hopefully sorted out}
By the discussion above, we have
\begin{theorem}
\label{thm:rabinization}
For every \ka automaton $\K$: $\Lang(\K) = \Lang(\R(\K))$. 
\end{theorem}

\begin{example}
Let us determine the accepting pairs of the DRA on the right of Figure \ref{fig:ka-dra}. 
We have ${\it fail} = \{ t_2,t_7,t_8\}, {\it buy}(\mathbf{1}) = \emptyset, {\it succeed}(\mathbf{1}) = \{t_1,t_6\}$, and ${\it buy}(\mathbf{2})=\{t_5,t_8\}, {\it succeed}(\mathbf{2}) = \{t_4,t_6,t_7\}$.

It is easy to see that the runs accepted by the pair $P_1$ are those that 
take $t_2,t_7,t_8$ only finitely often, and visit $(\mathbf{1}, \bot)$ 
infinitely often. They are accepted at rank $\mathbf{1}$.
The runs accepted at rank $\mathbf{2}$ are those accepted by $P_2$ but not by $P_1$.
They take $t_1, t_2, t_5, t_6, t_7, t_8$ finitely often, and so they are exactly 
the runs with a $t_4^\omega$ suffix. 
\end{example}

\subsection{The Automaton $\R(\varphi)$}

Given a $\G$-free formula $\varphi$, we define $\R(\varphi) = \R(\K(\varphi))$. 
By Theorem \ref{thm:kretcorrect} and Theorem \ref{thm:rabinization}, we 
have $\Lang(\R(\varphi)) = \Lang(\F\G\varphi)$. 

%Recall that the states of $\K(\varphi)$ are the formulae of $\Reach(\varphi)$. 
%After reading a finite word $w$, there is a distribution of tokens on the 
%non-final states of $\K(\varphi)$. After reading the same word, the state 
%of $\R(\varphi)$ is the seniority ranking of this distribution, and so a 
%partial ranking $r$ of $\Reach(\varphi)$.

If $w$ is accepted by $\R(\varphi)$ at rank $\rank{w}$, then we not only know 
that $w$ satisfies $\F\G\varphi$. In order to explain exactly what 
else we know, we need the following definition.

\begin{definition}
Let $\delta_{\R}$ be the transition function of the DRA $\R(\varphi)$ and 
let $w \in \Lang(\varphi)$ be a word. For every $j \in \mathbb{N}$, we denote 
by $\rankf{w_{0j}}$ the conjunction of the formulae of rank younger than or equal to $\rank{w}$ 
at the state $\delta_{\R}(i_\R, w_{0j})$.
\end{definition}

Intuitively, we also know that $w_j$ satisfies $\rankf{w_{0j}}$
for almost every index $j\in \mathbb{N}$, a fact we will use for the 
accepting condition of the Rabin automaton for general formulae in
Section \ref{sec:gen}. Before proving this, we give an example.

\begin{example}
Consider the Rabin automaton on the right of Figure \ref{fig:ka-dra}. Let $w = (\{b\}\{c\})^\omega$. Its corresponding run is $(t_3t_6)^\omega$, 
which is accepted at rank $\mathbf{1}$. For every even value $j$, 
$\rankf{w_{0j}}$ is the conjunction of the formulae of rank $\mathbf{1}$ and $\mathbf{2}$
at the state $(\mathbf{2},\mathbf{1})$. So we get $\rankf{w_{0j}} = (a \vee (b \; \U \; c)) \wedge
(b \; \U \; c) \equiv_p (b \; \U \; c)$, and therefore we know that infinitely many suffixes of $w$ satisfy
$(b \; \U \; c)$. In other words, the automaton tells us not only that 
$w \models \F\G(a \vee (b \; \U \; c))$, but also that $w \models \F\G(b \; \U \; c)$.
\end{example}

We now show this formally. %start the proof.
If $w \models \F\G\varphi$, there is a smallest index 
$\ind{w}{\varphi}$ at which $\varphi$ ``starts to hold''. 
For every index $j \geq \ind{w}{\varphi}$, we have
$w_j\models\bigwedge_{k = \ind{w}{\varphi}}^j \aft(\varphi, w_{kj}) \ . $ 
Intuitively, this formula is the conjunction
of the formulae ``tracked'' by the tokens of $\K(\varphi)$ born on days 
$\ind{w}{\varphi}, \ind{w}{\varphi}+1, \ldots, j$. 
These are the ``true'' tokens of $\K(\varphi)$, that is, those that 
eventually reach an accepting state. We get: %\todo{Jan: ``we'll use it in acc.cond.'' or the like}

\begin{lemma}
\label{lem:interpret}
Let $\varphi$ be a $\G$-free formula and let $w \in \Lang(\R(\varphi))$.
Then 
\begin{itemize}
\item[(1)] $\rankf{w_{0j}} \equiv \bigwedge_{k = \ind{w}{\varphi}}^j \aft(\varphi, w_{kj})$ for almost every $j \in \mathbb{N}$; and 
\item[(2)] $w_j \models \rankf{w_{0j}}$ for almost every $j \in \mathbb{N}$.
\end{itemize}
\end{lemma}

\section{DRAs for arbitrary $\F\G$-formulae}
\label{sec:gen-FG}

We construct a DRA for an arbitrary formula $\F\G$-formula $\F\G\varphi$. It suffices to construct
a \ka automaton, and then apply the construction of Section \ref{subsec:dekret}.
We show that the \ka automaton can be defined compositionally, as a parallel composition
of \ka automata, one for each $\G$-subformula. 

\begin{definition}
Given a formula $\varphi$, we denote by $\gsf(\varphi)$ the set of $\G$-subformulae
of $\varphi$, i.e., the subformulae of $\varphi$ of the form $\G\psi$.
\end{definition}

\noindent More precisely, for every $\setG \subseteq \gsf(\F\G\varphi)$
and every $\G\psi \in \setG$, we construct a \ka automaton $\K(\psi, \setG)$. 
Automata $\K(\psi, \setG)$ and $\K(\psi, \setG')$ for two different sets 
$\setG, \setG'$ have the same transition system, i.e., they differ only on the accepting condition. 
The automaton $\K(\psi, \setG)$ checks that $\F\G\psi$ holds, under the assumption
that $\F\G\psi'$ holds for all the subformulae $\G\psi'$ of $\psi$ that belong to $\setG$. 
Circularity is avoided, because automata for $\psi$ only rely on assumptions about proper 
subformulae of $\psi$. Loosely speaking, the Rabin automaton for $\F\G\varphi$
is the parallel composition (or product) of the Rabin automata for the $\K(\psi, \setG)$ (which are independent of
$\setG$), with an acceptance condition obtained from the acceptance
conditions of the $\K(\psi, \setG)$.

We only need to define the automaton $\K(\varphi, \setG)$, because the 
automata $\K(\psi, \setG)$ are defined inductively in exactly the same 
way. Intuitively, the automaton for $\K(\varphi, \setG)$ does not 
``track'' $\G$-subformulae of $\varphi$, it delegates that task to the automata
for its subformulae. This is formalized with the help of the following definition.

\begin{definition}
\label{def:aftg}
Let $\varphi$ be a formula and $\nu \in 2^{Ap}$. The formula
$\aftg(\varphi,\nu)$ is inductively defined as $\aft(\varphi,\nu)$, with only this difference: 
$$\hspace{2.0cm}\aftg(\G\varphi,\nu) =  \G\varphi \qquad \mbox{(instead of $\aft(\G\varphi,\nu) = \aft(\varphi,\nu)\wedge \G\varphi$).}$$
\noindent We define $\Reachg(\varphi)~=~\{ \aftg(\varphi,w) \mid w \in (2^{Ap})^* \}$
(up to $\equiv_p$).
\end{definition}

\begin{example}
\label{ex:aftg}
Let $\varphi = \psi  \U \neg a$, where $\psi =\G(a \wedge \X\neg a)$. We have 
$$\begin{array}{lclcl}
\aftg(\varphi, \{a\}) & = & \aftg(\psi,\{a\}) \wedge \varphi & \equiv_p & \psi \wedge \varphi \\
\aft(\varphi, \{a\}) & = & \aft(\psi,\{a\}) \wedge \varphi & \equiv_p & \neg a \wedge \psi \wedge \varphi  
\end{array}$$
\end{example}

\begin{definition}
\label{def:slave2}
Let $\varphi$ be a formula and let $\setG \subseteq \gsf(\varphi)$. 
The \ka automaton of $\varphi$ with respect to $\setG$ is
the quadruple $\K(\varphi,\setG)=$ $(\Reachg(\varphi), \varphi, \aftg, F_\setG)$, where  
$F_\setG$ contains the formulae $\varphi' \in \Reachg(\varphi)$ propositionally implied by $\setG$, i.e.\ the formulae 
satisfying $\bigwedge_{\G\psi \in \setG} \G\psi \models_p \varphi'$.
\end{definition}

\noindent Observe that only the set of accepting states of $\K(\varphi,\setG)$ depends on $\setG$. 
The following lemma shows that states reachable from final states are also final.

\begin{lemma}
Let $\varphi$ be a formula and let $\setG \subseteq \gsf(\varphi)$. 
For every $\varphi' \in \Reachg(\varphi)$, if $\bigwedge_{\G\psi \in \setG} \G\psi \models_p \varphi'$
then $\bigwedge_{\G\psi \in \setG} \G\psi \models_p \aftg(\varphi',\nu)$ for every $\nu \in 2^{Ap}$.
\end{lemma}
\begin{proof}
Follows easily from the definition of $\models_p$ and $\aftg(\G\psi) = \G\psi$.
\end{proof}

\begin{example}
Let $\varphi = (\G\psi) \U \neg a$, where $\psi =a \wedge \X\neg a$.
We have $\gsf(\varphi)=\{\G\psi\}$, and so two automata $\K(\varphi, \emptyset)$ and $\K(\varphi, \{\G\psi\})$,
whose common transition system is shown on the left of Figure \ref{fig:kanested}.
We have one single automaton $\K(\psi, \emptyset)$, shown on the right of the figure.
A formula $\varphi'$ is an accepting state of $\K(\psi, \emptyset)$ if $ \true \models_p \varphi'$;
and so the only accepting state of the automaton on the right is $\true$. 
On the other hand, $\K(\varphi, \{\G\psi\})$ has both $\G\psi$ and $\true$
as accepting states, but the only accepting state of
$\K(\varphi, \emptyset)$ is $\true$. 
\end{example}

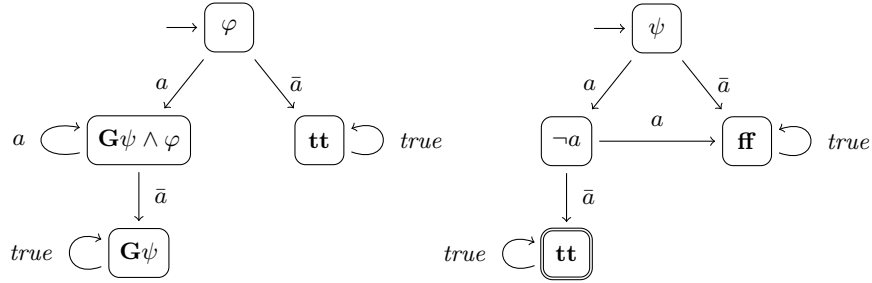
\begin{figure}
\myspace\myspace
\begin{center}
\begin{minipage}{5.0cm}
\begin{tikzpicture}[x=3cm,y=1.5cm,font=\footnotesize,initial text=,outer sep=1mm]
\tikzstyle{acc}=[double]
\node[state,initial] (a) at (0,0) {$\varphi$};
\node[state] (b) at (-0.4,-1) {$\G\psi \wedge \varphi$};
\node[state] (c) at (0.4,-1) {$\true$};
\node[state] (d) at (-0.4,-2) {$\G\psi$};

\path[->] (a) edge node[left]{$a$} (b) 
              edge node[right]{$\bar{a}$} (c)
          (b) edge node[right] {$\bar{a}$} (d)          
	  (b) edge[loop left, max distance=7mm,in=168,out=192,looseness=15] node[left]{$a$} (b)
          (c) edge[loop left, max distance=6mm,in=20,out=340,looseness=15] node[right]{${\it true}$} (c)
          (d) edge[loop left, max distance=6mm,in=160,out=200,looseness=15] node[left]{${\it true}$} (d)
;
;
\end{tikzpicture}
\end{minipage}
\qquad
\begin{minipage}{6.0cm}
\begin{tikzpicture}[x=3cm,y=1.5cm,font=\footnotesize,initial text=,outer sep=1mm]
\tikzstyle{acc}=[double]
\node[state,initial] (1) at (0,0) {$\psi$};
\node[state] (2) at (-0.4,-1) {$\neg a$};
\node[state] (3) at (0.4,-1) {$\false$};
\node[state,acc] (4) at (-0.4,-2) {$\true$};

\path[->] (1) edge node[left]{$a$} (2)
              edge node[right]{$\bar{a}$} (3)
          (2) edge node[above] {$a$} (3)
              edge node[right] {$\bar{a}$} (4)
          (3) edge[loop left, max distance=6mm,in=20,out=340,looseness=15] node[right]{${\it true}$} (3)
          (4) edge[loop left, max distance=6mm,in=160,out=200,looseness=15] node[left]{${\it true}$} (4)
;
;
\end{tikzpicture}
\end{minipage}
\end{center}
\myspace
\caption{\ka automata for $\varphi = (\G\psi) \; \U \neg a$, where  $\psi =a \wedge \X\neg a$.}
\label{fig:kanested}
\myspace\myspace
\end{figure}

\begin{theorem}
\label{thm:slaves}
Let $\varphi$ be a formula and let $w$ be a word. Then $w \models \F\G\varphi$ if{}f 
there is $\setG \subseteq \gsf(\varphi)$ such that (1) $w \in \Lang(\K(\varphi,\setG))$, and
(2) $w \models \F\G\psi$ for every $\G\psi \in \setG$.
\end{theorem}

Using induction on the structure of $\G$-subformulae we obtain:

\begin{theorem}
\label{thm:decomp}
Let $\varphi$ be a formula and let $w$ be a word. Then 
$w \models \F\G\varphi$ if{}f there is $\setG \subseteq \gsf(\F\G\varphi)$ such that
$w \in \Lang(\K(\psi,\setG))$ for every $\G\psi \in \setG$.
\end{theorem}

\subsection{The Product Automaton}\label{ssec:product}

Theorem \ref{thm:decomp} allows us to construct a Rabin automaton for an arbitrary formula of the form $\F\G\varphi$.
For every $\G\psi \in \gsf(\F\G\varphi)$ and every $\setG \subseteq \gsf(\F\G\varphi)$
let $\R(\psi,\setG) =(Q_\psi, i_\psi, \delta_\psi, {\it Acc}_\psi^\setG)$
be the Rabin automaton obtained by applying Definition \ref{def:rabinization} 
to the \ka automaton $\K(\psi,\setG)$. Since $Q_\psi, i_\psi, \delta_\psi$ do not depend on $\setG$,
we define the product automaton $\P(\varphi)$ as

$${\cal P}(\varphi)= \displaystyle \left(
\prod_{\G\psi \in \gsf(\varphi)} Q_\psi,\;
\prod_{\G\psi \in \gsf(\varphi)} \{i_\psi\}, \; 
\prod_{\G\psi \in \gsf(\varphi)} \delta_\psi,\; \; 
\bigvee_{\setG\subseteq \gsf{\varphi}} \bigwedge_{\G\psi \in \gsf(\varphi)} {\it Acc}_\psi^\setG
\right)
$$ 

Since each of the ${\it Acc}_\psi^\setG$ is a Rabin condition, we obtain a generalized Rabin condition.
This automaton can then be transformed into an equivalent Rabin automaton \cite{cav12}. However, as shown in
\cite{cav13}, for many applications it is better to keep it in this form. By Theorem \ref{thm:decomp}
we immediately get: 

\begin{theorem}
\label{thm:decomp2}
Let $\varphi$ be a formula and let $w$ be a word. Then 
$w \models \F\G\varphi$ if{}f there is $\setG \subseteq \gsf(\F\G\varphi)$ such that
$w \in L(\P(\varphi))$.
\end{theorem}

\section{DRAs for Arbitrary Formulae}\label{sec:gen}

In order to explain the last step of our procedure, consider the following example.

\begin{example}
\label{ex:lastex}
Let $\varphi = b \wedge \X b \wedge \G\psi$, where $\psi = a \wedge \X (b \U c)$ and let ${\it Ap}=\{a,b,c\}$.
The \ka automaton $\K(\psi)$ is shown in the middle of Figure \ref{fig:master}. Its corresponding
Rabin automaton $\R(\psi)$ is shown on the right, where the state $(\mathbf{i},\mathbf{j})$ indicates
that $\mathbf{\psi}$ has rank $\mathbf{i}$ and $b\U c$ has rank $\mathbf{j}$. We have
${\it fail} = \{t_1,t_5,t_6,t_7,t_8\}, {\it buy}(\mathbf{1})  =  \emptyset, {\it succeed}(\mathbf{1}) =  \{t_4,t_7\}$ and ${\it buy}(\mathbf{2})  =  \{t_3 \}, {\it succeed}(\mathbf{2})  =  \emptyset$.
% $$\begin{array}{rclcrclcrcl}
% {\it fail} &=& \{t_1,t_5,t_6,t_7,t_8\} & \qquad & {\it buy}(\mathbf{1}) & = & \emptyset & \qquad & {\it succeed}(\mathbf{1}) & = & \{t_4,t_7\} \\
%            &&&                                 & {\it buy}(\mathbf{2}) & = & \{t_3 \} & \qquad & {\it succeed}(\mathbf{2}) & = & \emptyset \\
% \end{array}$$\todo{changed, please check!}
\end{example}

Both $\K(\psi)$ and $\R(\psi)$ recognize $\Lang(\F\G\psi)$, but not $\Lang(\G\psi)$.
In particular, even though any word whose first letter does not contain $a$ can be immediately rejected, $\K(\psi)$
fails to capture this. This is a general problem of \ka automata: they can never ``reject (or accept) in finite time'' 
because the acceptance condition refers to an infinite number of tokens.

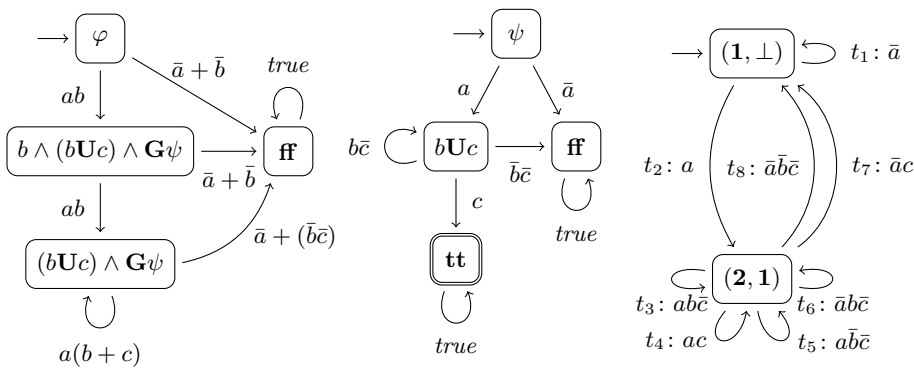
\begin{figure}
\myspace%\myspace\myspace
\begin{center}
\hspace*{-5cm}
\begin{minipage}{4cm}
\begin{tikzpicture}[x=2.5cm,y=1.5cm,font=\footnotesize,initial text=,outer sep=1mm]
\tikzstyle{acc}=[double]
\node[state,initial] (a) at (0,0) {$\varphi$};
\node[state] (b) at (0,-1) {$b \wedge (b \U c) \wedge \G\psi$};
\node[state] (c) at (0,-2) {$(b \U c) \wedge \G\psi$};
\node[state] (d) at (1,-1) {$\false$};

\path[->] (a) edge node[left]{$ab$} (b) 
              edge node[above]{\ $\bar{a}+\bar{b}$} (d)
          (b) edge node[left] {$ab$} (c)
              edge node[below] {$\bar{a}+\bar{b}$} (d)
          (c) edge [loop left, max distance=6mm,in=250,out=290,looseness=15] node[below]{$a(b+c)$} (c)
              edge [bend left=-30] node[right] {$\bar{a}+(\bar{b}\bar{c})$} (d)
(d) edge [loop above, max distance=6mm,in=110,out=70,looseness=15] node[above]{${\it true}$} (d)
;
;
\end{tikzpicture}
\end{minipage}
\quad
\begin{minipage}{3.5cm}
\begin{tikzpicture}[x=2cm,y=1.5cm,font=\footnotesize,initial text=,outer sep=1mm]
\tikzstyle{acc}=[double]
\node[state,initial] (1) at (0,0) {$\psi$};
\node[state] (2) at (-0.4,-1) {$b \U c$};
\node[state] (3) at (0.4,-1) {$\false$};
\node[state,acc] (4) at (-0.4,-2) {$\true$};
\node at (0,-3){};

\path[->] (1) edge node[left]{$a$} (2)
              edge node[right]{$\bar{a}$} (3)
          (2) edge node[below] {$\bar{b}\bar{c}$} (3)
              edge node[right] {$c$} (4)
          (2) edge[loop left, max distance=6mm,in=160,out=200,looseness=15] node[left]{$b \bar{c}$} (2)
(3) edge [loop below, max distance=6mm,in=290,out=250,looseness=15] node[below]{${\it true}$} (3)
(4) edge [loop below, max distance=6mm,in=290,out=250,looseness=15] node[below]{${\it true}$} (4)
;
;
\end{tikzpicture}
\end{minipage}
~
\begin{minipage}{4cm}
\begin{tikzpicture}[x=3cm,y=1.5cm,font=\footnotesize,initial text=,outer sep=1mm]
\tikzstyle{acc}=[double]
\node[state,initial] (1) at (0,0) {$(\mathbf{1},\bot)$};
\node[state] (2) at (0.0,-2.0) {$(\mathbf{2},\mathbf{1})$};

\path[->] (1) edge [bend left=-30] node[left]{$t_2\colon a$} (2)
          (2) edge [bend left=-55] node[right] {$t_7\colon \bar{a}c$} (1)
%{$\begin{array}{l}t_4 \colon \\\bar{a}c\end{array}$} (1)
              edge [bend left=-43] node[left] {$t_8 \colon \bar{a}\bar{b}\bar{c}$} (1)
          (1) edge[loop left, max distance=6mm,in=10,out=350,looseness=15] node[right]{$t_1 \colon \bar{a}$} (1)
          (2) edge[loop left, max distance=6mm,in=260,out=230,looseness=15] node[left]{$t_4 \colon ac$} (2)
          (2) edge[loop left, max distance=6mm,in=310,out=280,looseness=15] node[right]{$t_5 \colon a\bar{b}\bar{c}$} (2)
          (2) edge[loop left, max distance=6mm,in=190,out=170,looseness=15] node[below]{$t_3 \colon ab\bar{c}$} (2)
	  (2) edge[loop left, max distance=6mm,in=10,out=350,looseness=15] node[below]{$t_6 \colon \bar{a}b\bar{c}$} (2)
;
;
\end{tikzpicture}
\end{minipage}
\hspace*{-5cm}
\end{center}
\myspace
\caption{Automata $\M(\varphi)$, $\K(\psi)$, and $\R(\psi)$ for $\varphi = b \wedge \X b \wedge \G\psi$ and $\psi = a \wedge \X (b \U c)$.}
\label{fig:master}
\myspace\myspace\myspace
\end{figure}

\subsection{Master Transition System}

The ``accept/reject in finite time'' problem can be solved with the help of the {\em master transition system}
(an automaton without an accepting condition).

\begin{definition}
Let $\varphi$ be a formula. The {\em master transition system} for $\varphi$ is the
tuple $\M(\varphi) = (\Reach(\varphi), \varphi, \aft)$.
\end{definition}

\noindent The master transition system for the formula of Example \ref{ex:lastex} is shown on the left of 
Figure \ref{fig:master}. Whenever we enter state $\false$, we have $\aft(\varphi, w)=\false$ for the
word $w$ read so far, and so the run is not accepting.

Consider now the word $w = \{a,b,c\}^\omega$, which clearly satisfies $\varphi$. How do master $\M(\varphi)$
and slave $\K(\psi)$ decide together that $w \models \varphi$ holds? Intuitively, $\K(\psi)$ accepts, and tells the master
that $w \models \F\G\psi$ holds. The master reaches the state $(b \, \U \, c) \wedge \G\psi$ and stays there forever. 
Since she knows that $\F\G\psi$ holds, the master deduces that $w \models \varphi$ holds
if $w \models \F\G (b \, \U \, c)$. But where can it get this information from?

At this point the master resorts to Lemma \ref{lem:interpret}: the slave $\K(\psi)$ (or, more precisely, 
its Rabin automaton $\R(\psi)$) not only tells the master that $w$ satisfies $\F\G\psi$, but also at which rank, 
and so that $w_j$ satisfies $\rankf{w_{0j}}$ for almost every $j \in \mathbb{N}$. In our example,
during the run $w = \{a,b,c\}^\omega$, all tokens flow down the path $a \wedge \X(b\,\U \,c) \tran{a} b \, \U \, c \tran{c} \true$ ``in lockstep''.
No token buys any other, and all tokens of rank $\mathbf{1}$ succeed. The corresponding run of
$\R(\psi)$ executes the sequence $t_2t_4^\omega$ of transitions, stays in $(\mathbf{2},\mathbf{1})$
forever, and accepts at rank $\mathbf{1}$. So we have
$\rankf{w_{0j}} = (b \, \U \, c) \wedge \psi$ for every $j \geq 0$, and therefore the slave tells the master
that $w_j \models (b \, \U \,c)$ for almost every $j \in \mathbb{N}$.

So in this example the information required by the master is precisely the additional information supplied by $\K(\psi)$ due to Lemma \ref{lem:interpret}. 
The next theorem shows that this is always the case. 

\begin{theorem}
\label{thm:master}
Let $\varphi$ be a formula and let $w$ be a word. Let $\setG$ 
be the set of formulae $\G\psi \in \gsf(\varphi)$ such that 
$w \models \F\G\psi$.   
We have $w \models \varphi$ if{}f for almost every $i \in \mathbb{N}$:
\begin{equation}
\bigg(\bigwedge_{\G\psi \in \setG} \big( \G\psi \; \wedge \; \rankf{\psi,w_{0i}}\big) \wedge\bigwedge_{\G\psi\in\gsf(\varphi)\setminus\setG}\neg \G\psi \bigg) \models_p \aft(\varphi, w_{0i}) \ . \tag{$*$} 
\end{equation}
\end{theorem}\myspaceb

The automaton recognizing $\varphi$ is a product of the automaton $\P(\varphi)$ defined in Section \ref{ssec:product}, 
and $\M(\varphi)$. The run of $\P(\varphi)$ of a word $w$ 
determines the set $\setG \subseteq \gsf(\varphi)$ such that $w \models \F\G\psi$ if{}f 
$\psi \in \setG$. Moreover, each component of $\P(\varphi)$ accepts at a certain rank, 
and this determines the formula $\rankf{\psi,w_{0i}}$ for every $i \geq 0$
(it suffices to look at the state reached by the component of $\P(\varphi)$
in charge of the formula $\psi$). By Theorem \ref{thm:master},
it remains to check whether 
$(*)$ eventually 
%$$\bigwedge_{\G\psi \in \setG} \big( \G\psi \; \wedge \; \rankf{\psi,w_{0i}}\big) \models_p \aft(\varphi, w_{0i})$$\myspaceb
%\noindent 
holds. This is done with the help of $\M(\varphi)$, which ``tracks'' 
$\aft(\varphi, w_{0i})$. To check the property, we turn the accepting condition into a disjunction
not only on the possible $\setG \subseteq \gsf(\varphi)$, but also on the possible
rankings that assign to each formula $\G\psi \in \setG$ a rank. This corresponds to letting the
product guess which $\G$-subformulae will hold, and at which rank they will be accepted. The slaves check 
the guess, and the master checks that it eventually only visits states implied by the guess. 

\myspaceb

\subsection{The GDRA $\A(\varphi)$}

We can now formally define the final automaton $\A(\varphi)$ recognizing $\varphi$.
Let $\P(\varphi) = (Q_\P, i_\P, \delta_\P, {\it Acc}_\P)$ be the product automaton 
described in Section \ref{ssec:product},
and let $\M(\varphi) = (\Reach(\varphi), \varphi, \aft)$.
We let $$\A(\varphi) = (\Reach(\varphi) \times Q_P, (\varphi, i_\P), \aft \times \delta_P, {\it Acc})$$
where the accepting condition 
${\it Acc}$ is defined top-down as follows:
\begin{itemize}
\item ${\it Acc}$ is a disjunction containing a disjunct 
${\it Acc}^\setG_\pi$ for each pair $(\setG, \pi)$, 
where $\setG \subseteq \gsf(\varphi)$ and
$\pi$ is a mapping assigning to each $\psi \in \setG$ a
rank, i.e., a number between $\mathbf{1}$ and the
number of Rabin pairs of $\R(\varphi,\setG)$.
\item The disjunct ${\it Acc}^\setG_\pi$ is a conjunction of the form 
$\displaystyle {\it Acc}^\setG_\pi = {\it M}^\setG_\pi \wedge  \bigwedge_{\psi \in \setG}{\it Acc}_\pi(\psi)$.
\item Condition ${\it Acc}_\pi(\psi)$ states that $\R(\psi,\setG)$
accepts with rank $\pi(\psi)$ for every $\psi \in \setG$. 
It is therefore a Rabin condition with only one Rabin pair.
\item Condition ${\it M}^\setG_\pi$ states that $\A(\varphi)$ eventually stays within a subset
$F$ of states defined as follows. Let 
$(\varphi', r_{\psi_1}, \ldots, r_{\psi_k}) \in \Reach(\varphi) \times Q_P$, where $r_{\psi}$ is a ranking 
of the formulae of $\Reachg(\psi)$ for every $\G\psi \in \gsf(\varphi)$, and let $\mathcal F(r_\psi)$ 
be the conjunction of the states of $\K(\psi)$ to which $r_{\psi}$ assigns rank $\pi(\psi)$ or higher. Then
$$(\varphi', r_{\psi_1}, \ldots, r_{\psi_k}) \in F \ \mbox{ if{}f }\ 
\bigg( \bigwedge_{\G\psi \in \setG} \big(\G\psi \wedge \mathcal F(r_\psi)\big) \wedge\bigwedge_{\G\psi\in\gsf(\varphi)\setminus\setG}\neg \G\psi \bigg) \models_p \varphi'  . $$
Notice that ${\it M}^\setG_\pi$ is a co-B\"uchi condition, and so a Rabin condition with only one pair.
\end{itemize}

\begin{theorem}
\label{thm:finalthm}
For any LTL formula $\varphi$, $\Lang(\A(\varphi))=\Lang(\varphi)$. 
\end{theorem}

\section{The Alternation-Free Linear-Time $\mu$-calculus}
The linear-time $\mu$-calculus is a linear-time logic
with the same expressive power as B\"uchi automata and DRAs (see e.g. \cite{DBLP:conf/popl/Vardi88,DBLP:conf/fsttcs/Dam92}.
It extends propositional logic with the next operator
$\X$, and least and greatest fixpoints. This section 
is addressed to readers familiar with this logic.  We take as syntax
$$\varphi::=  \true \mid \false \mid a \mid \neg a \mid y \mid  \varphi\wedge\varphi \mid \varphi\vee\varphi \mid \X\varphi \mid \mu x.\varphi \mid \nu x. \varphi $$
\noindent where $y$ ranges over a set of variables. We assume that if $\sigma y. \varphi$ and $\sigma z. \psi$ are distinct subformulae of a formula, then
$y$ and $z$ are also distinct. A formula is {\em alternation-free} if 
for every subformula $\mu y. \varphi$ ($\nu y. \varphi$)
no path of the syntax tree leading from $\mu y$ ($\nu y$) to $y$ 
contains an occurrence of $\nu z$ ($\mu z$) for some variable $z$. 
For instance, $\mu y. (a \vee \mu z. (y \vee \X z)$ is alternation-free,
but $\nu y.\mu z ((a \wedge y) \vee \X z)$ is not. It is well known that
the alternation-free fragment is strictly more expressive than LTL 
and strictly less expressive than the full linear-time $\mu$-calculus. 
In particular, the property ``$a$ holds at every even moment'' is not 
expressible in LTL, but corresponds to $\nu y. (a \wedge \X\X y)$.

Our technique extends to the alternation-free linear-time
$\mu$-calculus. We have refrained from presenting 
it for this more general logic because it is less well known 
and formulae are more difficult to read. 
We only need to change the definition of the 
functions $\aft$ and $\aftg$. For the common part of the syntax
(everything but the fixpoint formulae) the definition is identical. For the 
rest we define
$$\begin{array}[t]{rcl}
\aft(\mu y.\varphi,\nu) &= &  \aft(\varphi,\nu)\vee \mu y.\varphi\\
\aft(\nu y.\varphi,\nu) &= &  \aft(\varphi,\nu)\wedge \nu y .\varphi\\
\end{array}
\qquad\qquad
\begin{array}[t]{rcl}
\aftg(\mu y.\varphi,\nu)&= &  \aftg(\varphi,\nu)\vee \mu y.\varphi\\
\aftg(\nu y.\varphi,\nu)&= &  \nu y.\varphi
\end{array}
$$
The automaton $\A(\varphi)$ is a product of automata, one for every
$\nu$-subformula of $\varphi$, and a master transition system.
Our constructions can be reused, and the proofs
require only technical changes in the structural inductions.

\section{Experimental results}

We compare the performance of the following tools and methods:
\begin{enumerate}
 \item[(T1)] ltl2dstar \cite{ltl2dstar} implements and optimizes \cite{DBLP:conf/wia/KleinB07} Safra's construction \cite{DBLP:conf/focs/Safra88}. It uses LTL2BA \cite{ltl2ba} to obtain the non-deterministic B\"uchi automata (NBA) first. Other translators to NBA may also be used, such as Spot \cite{spot-atva13} or LTL3BA \cite{ltl3ba} and in some cases may yield better results (see \cite{DBLP:conf/lpar/BlahoudekKS13} for comparison thereof), but LTL2BA is recommended by ltl2dstar and is used this way in PRISM \cite{prism}.
 \item[(T2)] Rabinizer \cite{atva12} and Rabinizer 2 \cite{atva13} implement a direct construction based on \cite{cav12} for fragments LTL$(\F,\G)$ and LTL$_{\setminus\G\U}$, respectively. The latter is used only on formulae not in LTL$(\F,\G)$.
 \item[(T3)] LTL3DRA \cite{DBLP:conf/atva/BabiakBKS13} which implements a construction via alternating automata, which is ``inspired by \cite{cav12}'' (quoted from \cite{DBLP:conf/atva/BabiakBKS13}) and performs several optimizations.
 \item[(T4)] Our new construction. Notice that we produce a state space with a logical structure,
which permits many optimizations; for instance, one could incorporate
the suspension optimization of LTL3BA \cite{DBLP:conf/spin/BabiakBDKS13}. 
However, in our prototype implementation we use only the following optimization:
In each state we only keep track of the slaves for formulae $\psi$ that are still ``relevant'' for 
the master's state $\varphi$, i.e. $\varphi[\psi/\true]\not\equiv_p\varphi[\psi/\false]$. 
For instance, after reading $\emptyset$ in $\G\F a\vee (b\wedge \G\F c)$, it is no longer
interesting to track if $c$ occurs infinitely often.
% However, in our prototype implementation we use only two optimizations:
% \begin{itemize}
%  \item In each state we only keep track of the slaves for formulae $\psi$ that are still ``relevant'' for 
% the master's state $\varphi$, i.e. $\varphi[\psi/\true]\not\equiv_p\varphi[\psi/\false]$. 
% For instance, after reading $\emptyset$ in $\G\F a\vee (b\wedge \G\F c)$, it is no longer
% interesting to track if $c$ occurs infinitely often.
%  \item Since the choice of the initial rankings of slaves does not affect acceptance, we start 
% in rankings that will occur repeteadly, omitting unnecessary initial transient parts of $\A(\varphi)$.
% \end{itemize}
\end{enumerate}
 
\begin{table}[t]
\myspace
{\footnotesize
$$\begin{array}{|l|r|r|r|r|}
  \hline \multicolumn{1}{|c|}{\text{Formula}}	&	 \multicolumn{1}{c|}{\text{T1}}	&	\text{T2}	&       \text{T3}	&	\text{T4}					\\	\hline
%&\text{DRA}&\text{DRA}&\text{tGDRA}&\text{tGDRA}\\
\F\G a\vee\G\F b 	&	4	&	4	&	1	&	1					\\	
(\F\G a\vee\G\F b)\wedge(\F\G c\vee\G\F d) 	&	11\,324	&	18	&	1	&	1					\\	
\bigwedge_{i=1}^3(\G\F a_i\rightarrow\G\F b_i)	&	1\,304\,706	&	462	&	1	&	1					\\	
%(\bigwedge_{i=1}^5\G\F a_i)\rightarrow\G\F b 	&	?	&	64	&	1	&	1					\\	
\bigwedge_{i=1}^2(\G\F a_i\rightarrow\G\F a_{i+1})	&	572	&	11	&	1	&	1					\\	
\bigwedge_{i=1}^3(\G\F a_i\rightarrow\G\F a_{i+1})	&	290\,046	&	52	&	1	&	1					\\	\hline
(\X(\G r \vee r \U (r \wedge s \U p)))\U(\G r \vee r \U (r\wedge s))	&	18	&	9	&	8	&	8					\\	
p \U(q\wedge \X(r\wedge (\F(s\wedge \X(\F(t\wedge \X(\F(u\wedge \X \F v))))))))	&	9	&	13&	13	&	13					\\	
(\G \F (a \wedge \X \X b) \vee \F \G b) \wedge \F \G (c \vee (\X a \wedge \X \X b))	&	353	&	73	&	-	&	12					\\	
\G \F (\X \X \X a \wedge \X \X \X \X b) \wedge \G \F (b \vee \X c) \wedge \G \F (c \wedge \X \X a)	&	2\,127	&	169	&	-	&	16					\\	
(\G \F a \vee \F \G b) \wedge (\G \F c \vee \F \G (d \vee \X e))	&	18\,176	&	80	&	-	&	2					\\	
(\G \F (a \wedge \X \X c) \vee \F \G b) \wedge (\G \F c \vee \F \G (d \vee \X a \wedge \X \X b))	&	?	&	142	&	-	&	12					\\	
a \U b \wedge (\G \F a \vee \F \G b) \wedge (\G \F c \vee \F \G d)\vee 	&	640\,771	&	210	&	8	&	7					\\
\qquad\qquad\vee a \U c \wedge (\G \F a \vee \F \G d) \wedge (\G \F c \vee \F \G b)		&&&&\\\hline
\F \G ((a \wedge \X \X b \wedge \G \F b)\U(\G(\X \X ! c \vee \X \X (a\wedge b))))	&	2\,053	&	-	&	-	&	11					\\	
\G(\F !a \wedge \F (b \wedge \X !c ) \wedge \G \F (a \U d)) \wedge \G \F ((\X d)\U(b \vee \G c))	&	283	&	-	&	-	&	7	\\					\hline
\varphi_{35}:\text{2 cause-1 effect precedence chain}&	6	&	-	&	-	&	6					\\	
\varphi_{40}:\text{1 cause-2 effect precedence chain}&	314	&	-	&	-	&	32					\\	
\varphi_{45}:\text{2 stimulus-1 response chain}&	1\,450	&	-	&	-	&	78					\\	
\varphi_{50}:\text{1 stimulus-2 response chain}&	28	&	-	&	-	&	23						\\
\hline%\\\varphi_{55}:\text{1-2 response chain constrained by a single proposition}	55&	28	&	-	&	-	&	23					
\end{array}$$
}
\caption{Some experimental results}
\label{table:exp}
\myspace\myspace
\end{table}

Table \ref{table:exp} compares these four tools.
For T1 and T2 we produce DRAs (although Rabinizer 2 can also produce GDRAs).
For T3 and T4 we produce GDRAs with transition acceptance (tGDRAs), which
can be directly used for probabilistic model checking without blow-up \cite{cav13}.
The table shows experimental results on four sets of formulae (see the four parts of the table)
\begin{enumerate}
 \item Formulae of the LTL$(\F,\G)$ fragment taken from (i) BEEM (BEnchmarks for Explicit Model checkers) \cite{beem} and from~\cite{DBLP:conf/cav/SomenziB00} on which ltl2dstar was originally tested~\cite{DBLP:journals/tcs/KleinB06}
(see Appendix~\ref{app:exp}); and (ii) fairness-like formulae. All the formulae were used already in \cite{cav12,DBLP:conf/atva/BabiakBKS13}. Our method usually achieves the same results as the optimized LTL3DRA, outperforming the first two approaches.
 \item Formulae of LTL$_{\setminus\G\U}$ taken from \cite{atva13} and \cite{DBLP:conf/concur/EtessamiH00}. They 
illustrate the problems of the standard approach to handle (i) $\X$ operators inside the scope of other temporal operators and (ii) conjunctions of liveness properties.
 \item Some further formulae illustrating the same phenomenon.
 \item Some complex LTL formulae expressing ``after Q until R'' properties%of \cite{DBLP:conf/cav/SomenziB00} and \cite{DBLP:conf/concur/EtessamiH00} and more complex formulae 
, taken from \textsc{Spec Pattern} \cite{DBLP:conf/icse/DwyerAC99} (available at \cite{specpatterns}) .
\end{enumerate}

All automata were constructed within a few seconds, with the exception of 
the larger automata generated by ltl2dstar:
it took several minutes for automata over ten thousand states and hours for hundreds of thousands of states. 
The automaton for $\bigwedge_{i=1}^3(\G\F a_i\rightarrow\G\F b_i)$ took even more than a day and ? denotes a time-out after one day.
Not applicability of the tool to the formula is denoted by $-$.
Additional details and more experimental results can be 
found in Appendix~\ref{app:exp}.

%Appendix and at \cite{web}.

\myspace

\section{Conclusions}

We have presented the first direct translation from LTL formulae to deterministic Rabin 
automata able to handle arbitrary formulae. The construction generalizes previous ones for LTL fragments \cite{cav12,atva12,atva13}. 
Given $\varphi$, we compute 
(1) the master, the slaves for each $\G\psi\in \gsf(\varphi)$, and their parallel composition, and 
(2) the acceptance condition: we first guess $\setG \subseteq \gsf(\varphi)$ which are true (this yields the accepting states of slaves), and then guess the ranks (this yields the %Lemma-\ref{lem:interpret} 
information for the master's co-B\"uchi acceptance condition).

The compositional approach opens the door to many possible optimizations. Since slave automata are
typically very small, we can aggressively try to optimize them, knowing that each reduced state 
in one slave potentially leads to large savings in the final number of states
of the product. So far we have only implemented the simplest optimizations, and we think there 
is still much room for improvement. 

We have conducted a detailed experimental comparison. Our construction 
outperforms two-step approaches that first translate the formula into a B\"uchi automaton and 
then apply Safra's construction. Moreover, despite handling full LTL, 
it is at least as efficient as previous constructions for fragments. Finally, we
produce a (often much smaller) generalized Rabin automaton, which can be directly 
used for verification, without a further translation into a standard Rabin automaton.

\bibliographystyle{alpha} % plain or alpha or splncs
\bibliography{refs}

\newpage
\appendix
\newcounter{applemma}
\setcounter{applemma}{1}
\newtheorem{alemma}[applemma]{Lemma}
\renewcommand\theapplemma{\Alph{applemma}}%{\Alph{theorem}}

\section{Proofs of Section~\ref{sec:rabinization}}

\begin{reflemma}{lem:aux}
Let $\varphi$ be a $\G$-free formula and let $w$ be a word. Then  $w \models \varphi$
if{}f $\aft(\varphi, w_{0i}) = \true$  for some $i \in \mathbb{N}$. 
\end{reflemma}
\begin{proof}
($\Leftarrow$): Follows directly from Lemma \ref{lem:fundaft}.

\noindent ($\Rightarrow$): Assume $w \models \varphi$. We proceed by induction on $\varphi$. We consider only some cases, the others are analogous.

\noindent {\bf Case} $\varphi = a$. Since $w \models \varphi$, we have $a \in w[0]$ and so $\aft(\varphi, w_{00})=\true$.

\noindent {\bf Case} $\varphi = \psi' \wedge \psi''$. Then $w \models \psi'$ and $w \models \psi''$. By induction hypothesis
there are $i', i''$ such that $\aft(\psi', w_{0i'})=\aft(\psi'', w_{0i''}) = \true$. Let $i = \max\{i',i''\}$. 
By the definition of $\aft()$ we have $\aft(\psi', w_{0i})=\aft(\psi'', w_{0i}) = \true$. 
Since we work up to propositional equivalence, we get $\aft(\psi' \wedge \psi'', w_{0i}) = \true \wedge \true = 
\true$.

\noindent {\bf Case} $\varphi = \F \varphi'$. Then there is $i \in \mathbb{N}$ such that $w_{i} \models \varphi'$.
By induction hypothesis there is $j \geq i$ such that $\aft(\varphi', w_{ij}) = \true$. By the definition of
$\aft()$ we have 
\begin{eqnarray*}
\aft(\varphi, w_{0j}) & = & \bigvee_{k=0}^{j} \aft(\varphi', w_{kj}) \vee \F\varphi' \\
& = & \left(\bigvee_{k=0}^{i-1} \aft(\varphi', w_{kj})\right) \vee \true \vee \left(\bigvee_{k=i+1}^{j} \aft(\varphi', w_{kj})\right) \vee \F\varphi' \\
& = & \true
\end{eqnarray*}

\noindent {\bf Case} $\varphi = \varphi' \U \varphi''$.  Then there is $i \in \mathbb{N}$ 
such that $w_{i} \models \varphi''$ and $w_{0}, w_1, \ldots, w_{i-1} \models \varphi'$. By induction
hypothesis, there are $j, j_0, j_1, \ldots, j_{i-1}$ such that
$\aft(\varphi'', w_{ij}) = \aft(\varphi', w_{0,j_0}) = \cdots = \aft(\varphi', w_{i-1,j_{i-1}})=\true$.
Let $J = \max\{j, j_0, \ldots, j_{i-1} \}$. 
By the definition of $\aft()$ we have 
$$\aft(\varphi'', w_{iJ}) = \aft(\varphi', w_{0J}) = \cdots = \aft(\varphi', w_{(i-1)J})=\true \ .$$ 
\noindent In particular, we get

\begin{align*}
      & \aft(\varphi' \U \varphi'', w_{0J})  \\
 = \; & \aft(\varphi'', w_{0J}) \vee (\aft(\varphi', w_{0J}) \wedge  \aft(\varphi' \U \varphi'', w_{1J})) \\
 = \; & \aft(\varphi'', w_{0J}) \vee \aft(\varphi' \U \varphi'', w_{1J}) \\
 = \; & \cdots \\
 = \; & \left(\bigvee_{k=0}^{i} \aft(\varphi'', w_{kJ}) \right) \vee \aft(\varphi' \U \varphi'', w_{(i+1)J}) \\
 = \; & \left(\bigvee_{k=0}^{i-1} \aft(\varphi'', w_{kJ}) \right) \vee \true \vee \aft(\varphi' \U \varphi'', w_{(i+1)J}) \\
 = \; & \true
\end{align*}
\qed\end{proof}

\begin{reftheorem}{thm:kretcorrect}
Let $\varphi$ be a $\G$-free formula. Then $\Lang(\K(\varphi)) = \Lang(\F\G\varphi)$.
\end{reftheorem}
\begin{proof}
$\Lang(\K(\varphi)) \subseteq  \Lang(\F\G\varphi)$. Let $w \in \Lang(\K(\varphi))$. By definition,
$\true$ is the only accepting state of $\K(\varphi)$. So, by the definition
of the acceptance condition of \ka automata, the $i$-th token of $\K(\varphi)$ eventually reaches
$\true$ for almost every $i \in \mathbb{N}$. Since after reading
$w_{0j}$ the $i$-th token of $\K(\varphi)$ is in the state $\aft(\varphi, w_{ij})$, 
for almost every $i \in \mathbb{N}$ there is 
$k \geq i$ such that $\aft(\varphi, w_{ik})= \true$.
By Lemma \ref{lem:fundaft}, we have $w_{i\infty} \models \varphi$ for almost every  $i \in \mathbb{N}$.
So $w \models \F\G\varphi$.

\noindent $\Lang(\F\G\varphi) \subseteq \Lang(\K(\varphi))$. Let  $w \in \Lang(\F\G\varphi)$. Then 
$w_{i} \models \varphi$ for almost every  $i \in \mathbb{N}$. By Lemma \ref{lem:aux}, for almost every  $i \in \mathbb{N}$ there is $j \geq i$ such that $\aft(\varphi, w_{1i}) =\true$. So almost every token of $\K(\varphi)$ 
eventually reaches the accepting state, and therefore $\K(\varphi)$ accepts.
\qed\end{proof}

We prove a slight generalization of Lemma \ref{lem:interpret}, valid
for arbitrary formulae, and not only for $\G$-free ones. The generalization is 
formulated in terms of the function $\aftg$ of Definition \ref{def:aftg}. 
Since $\aft$ and $\aftg$ coincide on $\G$-free formulae, the new formulation indeed generalizes the one in the main text. 
 
\begin{reflemma}{lem:interpret}
Let $\varphi$ be an arbitrary (not necessarily $\G$-free) formula and let $w \in \Lang(\R(\varphi))$.
Then 
\begin{itemize}
\item[(1)] $\rankf{\varphi,w_{0j}} = \displaystyle \bigwedge_{k = \ind{w}{\varphi}}^j \aftg(\varphi, w_{kj})$ for almost every $j \in \mathbb{N}$.
\item[(2)] $w_j \models \rankf{\varphi,w_{0j}}$ for almost every $j \in \mathbb{N}$.
\end{itemize}
\end{reflemma}
\begin{proof}
\noindent (1) Since $w \in \Lang(\R(\varphi))$, the word $w$ is accepted by $\R(\varphi)$ at
rank $\rank{w}$. So, by the definition of the acceptance condition, there is a point 
in time, say $i$, such that for every $j \geq i$, the $j$-th token neither fails nor is bought
by a token of rank older than $\rank{w}$ at a non-accepting state. So the $j$-th token
eventually reaches an accepting state. Then, for every $j \geq i$, and so for almost 
every $j$, a token is true at time $j$ if{}f its rank is younger than or equal to $\rank{w}$. 
So both  $\rankf{w_{0j}}$ and $\bigwedge_{k = \ind{w}{\varphi}}^j \aftg(\varphi, w_{kj})$
are equal to the conjunction of the formulae tracked by the true tokens.

\noindent (2) By Lemma~\ref{lem:fundaft}, we have $\displaystyle w_j \models \bigwedge_{k = \ind{w}{\varphi}}^j \aft(\varphi, w_{kj})$ for every $j \geq \ind{w}{\varphi}$, hence also $\displaystyle w_j \models \bigwedge_{k = \ind{w}{\varphi}}^j \aftg(\varphi, w_{kj})$ for every $j \geq \ind{w}{\varphi}$, as $\aftg$ is obtained from $\aft$ by possibly omitting some conjuncts. The claim then follows immediately from (1).
\qed\end{proof}

\newpage
\section{Proofs of Section~\ref{sec:gen-FG}}

To prove Theorem \ref{thm:slaves} we need a lemma, where we use the following notation:
given a set $F$ of formulae, we write $F \models_p \varphi$ as an abbreviation of $\bigwedge_{\psi \in F} \psi \models_p \varphi$. 

\begin{alemma}\label{lem:auxg}
Let $\varphi$ be a formula and let $w$ be a word. If $w \models \varphi$, then
there is $\setG \subseteq \gsf(\varphi)$ such that $w \models \F\G\psi$ for every $\G\psi \in \setG$
and $\setG \models_p \aftg(\varphi, w_{0i})$ for almost every 
$i \in \mathbb{N}$.
\end{alemma}
\begin{proof}

Assume $w \models \varphi$, and let $\setG$ be the set of formulae $\G\psi \in \gsf(\varphi)$ such that 
$w \models \F\G\psi$. We proceed by induction on the structure of $\varphi$. We consider only some cases, the others are either trivial or analogous.

\noindent {\bf Case} $\varphi = \varphi_1 \vee \varphi_2$. Then w.l.o.g. $w \models \varphi_1$. 
By induction hypothesis we have $\setG \models_p \aftg(\varphi_1, w_{0i})$ for almost every 
$i \in \mathbb{N}$, and since $\aftg(\varphi, w_{0i}) = \aftg(\varphi_1, w_{0i}) \vee \aftg(\varphi_2, w_{0i})$,
we get $  \setG \models_p \aftg(\varphi, w_{0i})$ for almost every $i \in \mathbb{N}$.

\noindent {\bf Case} $\varphi = \X \varphi'$. Then $w_1 \models \varphi'$. By induction hypothesis
we have $$\setG \models_p \aftg(\varphi', w_{1i}) = \aftg(\aftg(\X\varphi', w[0]), w_{1i}) = \aftg(\X\varphi', w_{0i}) \;. $$

\noindent {\bf Case} $\varphi = \varphi_1 \U \varphi_2$. Then $w_j \models \varphi_2$ for some 
 $j \in \mathbb{N}$, and $w_k \models \varphi_1$ for every $k < j$.
By induction hypothesis, for almost every $i \in \mathbb{N}$ we have $  \setG \models_p \aftg(\varphi_2, w_{ji})$, and 
for every $0 \leq k \leq j-1$ we have $\setG \models \aftg(\varphi_1, w_{ki})$. By the definition of $\aft()$, we get:
$$\begin{array}{rcl}
\setG & \models_p & \aftg(\varphi_2, w_{ji})  \models_p  \aftg(\varphi_1 \U \varphi_2, w_{ji})\\
\setG & \models_p & %\aftg(\varphi_2, w_{(j-1)i}) \vee (
\aftg(\varphi_1, w_{(j-1)i}) \wedge \aftg(\varphi_1 \U \varphi_2, w_{ji})%) %\\
\models_p % &  &  = \; 
\aftg(\varphi_1 \U \varphi_2, w_{(j-1)i}) \\
& \cdots & \\
\setG & \models_p & %\aftg(\varphi_2, w_{1i}) \vee (
\aftg(\varphi_1, w_{1i}) \wedge \aftg(\varphi_1 \U \varphi_2, w_{2i}) \models_p %) =  
\aftg(\varphi_1 \U \varphi_2, w_{1i}) \\
\setG & \models_p & %\aftg(\varphi_2, w_{0i}) \vee (
\aftg(\varphi_1, w_{0i}) \wedge \aftg(\varphi_1 \U \varphi_2, w_{1i}) \models_p %) =   
\aftg(\varphi_1 \U \varphi_2, w_{0i}) \;.
\end{array}$$

\noindent {\bf Case} $\varphi = \G\varphi'$. Then $w \models \G\varphi'$, and so $\G\varphi' \in \setG$ by the definition of $\setG$. 
In particular, we have  $ \setG \models_p \G\varphi'$, and, since
$\aftg(\varphi, w_{0i})=\G\varphi'$ for every index $i \in \mathbb{N}$, 
we get $ \setG  \models_p  \aftg(\varphi, w_{0i})$  for every $i \in \mathbb{N}$. 

\qed\end{proof}

\begin{reftheorem}{thm:slaves}
Let $\varphi$ be a formula and let $w$ be a word. Then $w \models \F\G\varphi$ if{}f 
there is $\setG \subseteq \gsf(\varphi)$ such that (1) $w \in \Lang(\K(\varphi,\setG))$, and
(2) $w \models \F\G\psi$ for every $\G\psi \in \setG$.
\end{reftheorem}
\begin{proof}

($\Rightarrow$): Let $\setG$ be the set of formulae $\G\psi$ such that  
$w \models \F\G\psi$. Then $\setG$ satisfies (2) by definition. To prove that
it satisfies (1), observe that, since $w \models \F\G\varphi$, we have $w_j \models \varphi$
for almost every $j$. So we can apply Lemma \ref{lem:auxg} to $w_j$, and get that for almost every
$j \in \mathbb{N}$ there is $i \geq j$ such that $\setG \models_p \aftg(\varphi,w_{ji})$. So, for almost all $j \in \mathbb{N}$, the $j$-th token of 
$\K(\varphi,\setG)$ reaches an accepting state $(i-j)$ steps later, and so $\Lang(\K(\varphi,\setG)$ accepts.

($\Leftarrow$): Assume $w \in \Lang(\K(\varphi,\setG))$ and 
$w \models \F\G\psi$ for every $\G\psi \in \setG$. Since $w \in \Lang(\K(\varphi,\setG))$,
for almost every $i \in \mathbb{N}$ there is $j \geq i$ such that $\setG \models_p \aftg(\varphi,w_{ij})$. 
Let $k_1 \in \mathbb{N}$ be large enough so that $w_i \models \G\psi$ for every $i \geq k_1$ and every $\G\psi \in \setG$. Let $k_2$ be large enough
so that for every $i \geq k_2$ there is $j \geq i$ such that $\setG \models_p \aftg(\varphi,w_{ij})$. 
Finally, let $k = \max\{k_1, k_2\}$, and let $i \geq k$ arbitrary. Since $i \geq k_2$, 
there is $j \geq i$ such that $\setG \models_p \aftg(\varphi,w_{ij})$; since $i \geq k_1$, we have $w_j \models \setG$. So we get 
$w_j \models \aftg(\varphi,w_{ij})$, which implies $w_i \models \varphi$. So we have  $w_i \models \varphi$ for every
$i \geq k$, and so $w \models \F\G \varphi$.
\qed\end{proof}

\begin{reftheorem}{thm:decomp}
Let $\varphi$ be a formula and let $w$ be a word. Then 
$w \models \F\G\varphi$ if{}f there is $\setG \subseteq \gsf(\F\G\varphi)$ such that
$w \in \Lang(\K(\psi,\setG))$ for every $\G\psi \in \setG$.
\end{reftheorem}
\begin{proof}
($\Rightarrow$): By repeated application of Theorem \ref{thm:slaves},  
we have $w \in \Lang(\K(\psi,\setG))$ for every $\G$-subformula $\G\psi$ of $\varphi$.

\noindent ($\Leftarrow$): Let $\setG' = \setG \setminus \{\varphi\}$. By Theorem \ref{thm:slaves} it suffices to show that 
$w \models \F\G\psi$ for every $\G\psi \in \setG'$. Let $\G\psi \in \setG'$. 
We proceed by structural induction on the subformula order. If  $\psi$ has no subformulae in $\setG'$, then $\Lang(\K(\psi,\setG))=\Lang(\K(\psi))$, and $w \models \F\G\psi$  follows from
 Theorem \ref{thm:kretcorrect}. Otherwise,
by induction hypothesis, $w \models \F\G\psi'$ for each subformula $\G\psi'$
of $\psi$ such that $\G\psi' \in \setG$. Together with $w \in \Lang(\K(\psi,\setG))$
and Theorem \ref{thm:slaves}, we get $w \models \F\G\psi$.
\qed\end{proof}

\section{Proofs of Section~\ref{sec:gen}}

We prove Theorem \ref{thm:master} with the help of a sequence of lemmas.
Recall that when $w \models \F\G\varphi$ we denote by $\ind{w}{\varphi}$ the
smallest index $j$ such that $w_j \models \G\varphi$. In the course of the proof
we abuse notation and write $\neg \G \varphi$ to denote the formula obtained as 
the formula in negation normal form obtained by pushing the negation inwards.

\begin{alemma}
\label{lem:master1}
Let $\varphi$ be a formula and let $w$ be a word.  Let $\setG$ 
be the set of formulae $\G\psi \in \gsf(\varphi)$ such that 
$w \models \F\G\psi$.  Then
$$\bigg(\aftg(\varphi, w_{0i}) \wedge \bigwedge_{\G\psi \in \setG} \big( \G\psi \; \wedge \; 
\bigwedge_{j=0}^i \aftg(\psi,w_{ji}) \big) \;  \wedge \bigwedge_{\G\psi\in\gsf(\varphi)\setminus\setG}\neg \G\psi\bigg)  \models_p \; \aft(\varphi, w_{0i})$$
\noindent holds for every $i \in \mathbb{N}$.
\end{alemma}
\begin{proof}
We introduce the abbreviation

$$\Omega(\varphi, w_{0i}) := \bigwedge_{\G\psi \in \setG} \big( \G\psi \; \wedge \;\bigwedge_{i=0}^j \aftg(\psi,w_{ji}) \big) \wedge \bigwedge_{\G\psi\in\gsf(\varphi)\setminus\setG}\neg \G\psi \ .$$
\noindent Observe that, for every $j \geq 0$, we have $w \models \F\G \psi$ 
if{}f $w_j \models \F\G \psi$. So the set $\setG$, which in principle depends
on $w$, in fact coincides for $w$ and $w_j$, and by the definition of $\Omega$ we get
\begin{equation}
\displaystyle \Omega(\varphi', w_{0i}) \models_p \Omega(\varphi', w_{ji}) \qquad \mbox{ for every $j \geq 0$}
\label{eq15}
\end{equation}

\noindent Making use of the abbtreviation, our goal is to prove that
\begin{equation*}
%\label{eq1}
\aftg(\varphi, w_{0i}) \wedge \Omega(\varphi, w_{0i}) \models_p \aft(\varphi, w_{0i})
\end{equation*}
\noindent holds for every $i \in \mathbb{N}$. We proceed by structural induction on $\varphi$. 
For this it is convenient to extend $\setG$ to a function that, given a formula $\varphi'$, 
returns the set $\setG(\varphi')$ of formulae $\G\psi \in \gsf(\varphi')$ such that $w \models \F\G\psi$. In particular, we have $\setG = \setG(\varphi)$.

We consider only some of the induction cases.

\vspace{0.2cm}
\noindent {\bf Case $\varphi = a$}. Follows from  $\aftg(a, w_{0i}) = \aft(a, w_{0i})$. 

\vspace{0.2cm}
\noindent {\bf Case $\varphi = \varphi_1 \vee \varphi_2$}. Then $\setG(\varphi) = \setG(\varphi_1) \cup \setG(\varphi_2)$. It follows
\begin{equation}
\label{eq01}
\Omega(\varphi, w_{0i})= \Omega(\varphi_1, w_{0i}) \wedge \Omega(\varphi_2, w_{0i})
\end{equation}
\noindent and therefore
$$\begin{array}{rclcl}
&   & \aftg(\varphi_1 \vee \varphi_2, w_{0i}) \wedge \Omega(\varphi, w_{0i}) \\[0.2cm]
& = & (\aftg(\varphi_1, w_{0i}) \vee \aftg(\varphi_2, w_{0i})) \wedge \Omega(\varphi_1, w_{0i}) \wedge
\Omega(\varphi_2, w_{0i}) & \hspace{0.2cm} & \mbox{by def. of $\aftg$, (\ref{eq01})} \\[0.2cm]
& \models_p & \aft(\varphi_1, w_{0i}) \vee \aft(\varphi_2, w_{0i}) & & \mbox{by ind. hyp.}\\[0.2cm]
& = & \aft(\varphi_1 \vee \varphi_2, w_{0i}) &  & \mbox{by def. of $\aft$}
\end{array}$$
 
\vspace{0.2cm}\noindent {\bf Case $\varphi = \F \varphi'$}. 
We collect some facts. First, by the definition of $\aft$ we have 
\begin{eqnarray}
\aft(\varphi, w_{0i}) & = & \bigvee_{j=0}^{i} \aft(\varphi', w_{ji}) \vee \F\varphi' \label{eq11}\\
\aftg(\varphi, w_{0i}) & = & \bigvee_{j=0}^{i} \aftg(\varphi', w_{ji}) \vee \F\varphi' \label{eq12}
\end{eqnarray}
\noindent Second, since $\F\varphi'$ and $\varphi'$ have the same $\G$-subformulae, we have 
$\setG(\F\varphi') = \setG(\varphi')$, and so 
\begin{equation}
\label{eq13}
\Omega(\F\varphi', w_{0i}) = \Omega(\varphi', w_{0i})
\end{equation}
\noindent Finally, applying the induction hypothesis to $\varphi'$ and $w_j$ for every $j \leq i$, we get
\begin{equation}
\label{eq14}
\aftg(\varphi', w_{ji}) \wedge \Omega(\varphi', w_{ji}) \models_p \aft(\varphi', w_{ji}) 
\qquad \mbox{for every $j \leq i$}
\end{equation}
Putting all these facts together, we get

$$\begin{array}{rclcl}
&   & \aftg(\F\varphi', w_{0i}) \wedge \Omega(\F\varphi', w_{0i}) \\[0.2cm]
& = & \left(\bigvee_{j=0}^{i} \aftg(\varphi', w_{ji}) \vee \F\varphi' \right) \wedge \Omega(\varphi', w_{0i}) & \hspace{0.5cm} & \mbox{by (\ref{eq12}), (\ref{eq13})} \\[0.2cm]
& \equiv_p & \left(\bigvee_{j=0}^{i} \aftg(\varphi', w_{ji}) \vee \F\varphi' \right) \wedge \; \bigwedge_{j=0}^i \Omega(\varphi', w_{ji}) & &\mbox{by (\ref{eq15})}\\[0.2cm]
& \models_p & \left(\bigvee_{j=0}^{i} \aft(\varphi', w_{ji})\right) \vee \F\varphi' & & \mbox{by (\ref{eq14})}\\[0.2cm]
& = & \aft(\F\varphi', w_{0i}) &  & \mbox{by (\ref{eq11})}
\end{array}$$

\noindent {\bf Case $\varphi = \G \varphi'$}. 
%The former cases still hold if 
%$\rankf{\psi,w_{0i}}$ is replaced by an arbitrary formula. The case $\varphi = \G \varphi'$
%is the only one where we need to resort to the definition of  $\rankf{\psi,w_{0i}}$.
By the definition of $\aft()$ we have 
\begin{eqnarray}
\aft(\G\varphi', w_{0i})  & = & \bigwedge_{k=0}^{i} \aft(\varphi', w_{ki}) \wedge \G\varphi' \label{eq1}\\
\aftg(\G\varphi', w_{0i}) & = & \G\varphi' \label{eq2}
\end{eqnarray}
\noindent We consider two cases.

\vspace{0.2cm}
\noindent {\bf Case $w \models \F\G\varphi'$}. Then $\setG(\G\varphi') = \setG(\varphi') \cup \{ \G \varphi'\}$, and so
\begin{equation}
\label{eq4}
\Omega(\G\varphi', w_{0i})  =  \Omega(\varphi', w_{0i}) \wedge \G\varphi' \wedge \bigwedge_{j=0}^{i} \aftg(\varphi', w_{ji}) 
\end{equation}
\noindent Now we get:
$$\begin{array}{rllcl}
&   & \aftg(\G\varphi', w_{0i}) \wedge \Omega(\G\varphi', w_{0i}) \\[0.2cm]
& = & \G \varphi' \wedge \left( \Omega(\varphi', w_{0i}) \wedge \G\varphi' \wedge \bigwedge_{j=0}^{i} \aftg(\varphi', w_{ji}) \right) & \hspace{0.5cm} & \mbox{by (\ref{eq2}), (\ref{eq4})} \\[0.2cm]
& \equiv_p & \G \varphi' \wedge \bigwedge_{j=0}^i \left( \Omega(\varphi', w_{ji}) \wedge \aftg(\varphi', w_{ji}) \right)& & \mbox{by (\ref{eq15})}\\[0.2cm]
& \models_p & \G \varphi' \wedge \bigwedge_{j=0}^i \aft(\varphi', w_{ji}) &  & \mbox{by ind. hyp.} \\[0.2cm]
& = & \aft(\G\varphi', w_{0i}) &  & \mbox{by (\ref{eq1})}
\end{array}$$

\noindent {\bf Case $w \not\models \F\G\varphi'$}. Then 
$\setG(\G\varphi') = \setG(\varphi')$, and therefore
\begin{equation}
\label{eq5}
\Omega(\G\varphi', w_{0i})  =  \Omega(\varphi', w_{0i}) \wedge \neg \G\varphi' 
\end{equation}

\noindent Now we get:
$$\begin{array}{rllcl}
&   & \aftg(\G\varphi', w_{0i}) \wedge \Omega(\G\varphi', w_{0i}) \\[0.2cm]
& = & \G \varphi' \wedge \left( \Omega(\varphi', w_{0i}) \wedge \neg \G\varphi'\right) & \hspace{0.5cm} & \mbox{by (\ref{eq2}), (\ref{eq5})} \\[0.2cm]
& \equiv_p & \false \\[0.2cm]
& \models_p & \aft(\G\varphi', w_{0i}) 
\end{array}$$
\noindent because $\false$ propositionally implies any formula. \qed\end{proof}

\begin{alemma}
\label{lem:master2}
Let $\varphi$ be a formula and let $w$ be a word such that $w \models \varphi$.  Let $\setG$ 
be the set of formulae $\G\psi \in \gsf(\varphi)$ such that 
$w \models \F\G\psi$.  Then 
$$\bigg(\bigwedge_{\G\psi \in \setG} \big( \G\psi \; \wedge \; \bigwedge_{j=0}^i \aftg(\psi,w_{ji}) \big) \; \wedge \bigwedge_{\G\psi\in\gsf(\varphi)\setminus\setG}\neg \G\psi \bigg)\;\; \models_p \;\; \aft(\varphi, w_{0i})$$
holds for almost every $i \in \mathbb{N}$.
\end{alemma}
\begin{proof}
By Lemma \ref{lem:master1} we have 
$$\bigg(\aftg(\varphi, w_{0i}) \wedge \bigwedge_{\G\psi \in \setG} \big( \G\psi \; \wedge \; \bigwedge_{j=0}^i \aftg(\psi,w_{ji}) \big) \; \wedge \bigwedge_{\G\psi\in\gsf(\varphi)\setminus\setG}\neg \G\psi \bigg) \models_p \aft(\varphi, w_{0i}) $$
\noindent for  every $i \geq 0$. By Lemma \ref{lem:auxg}, and since $w \models \varphi$, we get
$\bigwedge_{\G\psi \in \setG} \G\psi \models_p \aftg(\varphi, w_{0i})$ for almost
every $i \geq 0$, and the result follows. 
\end{proof}

\begin{alemma}
\label{lem:master3}
Let $\varphi$ be a formula and let $w$ be a word. Let $\setG$ 
be the set of formulae $\G\psi \in \gsf(\varphi)$ such that 
$w \models \F\G\psi$.  Then $w \models \varphi$ if{}f 
$$\bigg( \bigwedge_{\G\psi \in \setG} \big( \G\psi \; \wedge \; \bigwedge_{j=\ind{w}{\psi}}^i \aftg(\psi,w_{ji}) \big) \wedge \bigwedge_{\G\psi\in\gsf(\varphi)\setminus\setG}\neg \G\psi \bigg) \; \models_p \; \aft(\varphi, w_{0i})$$
holds for almost every $i \in \mathbb{N}$.
\end{alemma}
\begin{proof}

Let $k$ be the maximum of ${\it ind}(w,\psi)$
over all $\G\psi \in \setG$. In particular, we have $w_k \models \G\psi$ for every $\G\psi \in \setG$.

\noindent ($\Leftarrow$): Assume the expression holds for almost every $i \in \mathbb{N}$.
We first claim that
$$w_i \models \bigwedge_{\G\psi \in \setG} \big( \G\psi \; \wedge \; \bigwedge_{j=\ind{w}{\psi}}^i \aftg(\psi,w_{ji})\big) \;  \wedge \bigwedge_{\G\psi\in\gsf(\varphi)\setminus\setG}\neg \G\psi$$
\noindent holds for every $i \geq k$.

By the definition of $k$, we have $w_i \models \G\psi$ for every $\G\psi \in \setG$ and every $i \geq k$. Applying Lemma~\ref{lem:fundaft}, 
we get $w_{i} \models \aft(\psi, w_{ji})$ for every $\ind{w}{\psi} \leq j \leq i$, and so
$w_i \models \bigwedge_{j = \ind{w}{\psi}}^i \aft(\varphi, w_{ji})$. 
Since $\aftg$ is obtained from $\aft$ by possibly omitting some conjuncts, we also have 
$w_i \models \bigwedge_{j = \ind{w}{\psi}}^i \aftg(\varphi, w_{ji})$. Finally, for every
$\G\psi \in \gsf(\varphi)\setminus\setG$ we have $w \not\models \F\G\psi$, and so $w_i \models \neg \G\psi$.
This proves the claim.

By this claim and the assumption of the lemma, we have $w_i \models \aft(\varphi, w_{0i})$ for 
almost every $i \geq \mathbb{N}$, and
so $w \models \varphi$ by Lemma~\ref{lem:fundaft}.

\noindent ($\Rightarrow$): Assume $w \models \varphi$. By Lemma~\ref{lem:fundaft}
we have $w_k \models \aft(\varphi, w_{0k})$. Since $w_k \models \G\psi$ for every $\G\psi \in \setG$, 
we can apply Lemma \ref{lem:master2} to the formula $\aft(\varphi, w_{0k})$ and the word $w_k$, which by $\setG(\aft(\varphi, w_{0k}))=\setG(\varphi)=\setG$ yields that
$$\bigg(\bigwedge_{\G\psi \in \setG} \big( \G\psi \; \wedge \; \bigwedge_{j=k}^i \aftg(\psi,w_{ji}) \big) \; \wedge \bigwedge_{\G\psi\in\gsf(\varphi)\setminus\setG}\neg \G\psi \bigg)\;\; \models_p \;\; \aft(\aft(\varphi, w_{0k}), w_{ki})$$
\noindent holds for almost every $i \geq 0$. The result follows from $\aft(\aft(\varphi, w_{0k}), w_{ki}) = \aft(\varphi, w_{0i})$,
and the fact that $k \geq \ind{w}{\psi}$ for every $\G\psi \in \setG$.
\qed\end{proof}

Lemma \ref{lem:master3} is a purely logical statement concerning only the logic LTL, and independent of any
automata-theoretic considerations. In order to obtain Theorem \ref{thm:master} we use Lemma \ref{lem:interpret},
which interprets the conjunction $\bigwedge_{j=\ind{w}{\psi}}^i \aftg(\psi,w_{ji})$ as the
formula $\rankf{\psi,w_{0i}}$ given by the acceptance rank of the run of $\A(\psi)$ on $w$.

\begin{reftheorem}{thm:master}
Let $\varphi$ be a formula and let $w$ be a word. Let $\setG$ 
be the set of formulae $\G\psi \in \gsf(\varphi)$ such that 
$w \models \F\G\psi$. Then $w \models \varphi$ if{}f 
$$\bigg(\bigwedge_{\G\psi \in \setG} \big( \G\psi \; \wedge \; \rankf{\psi,w_{0i}}\big) \; \wedge \bigwedge_{\G\psi\in\gsf(\varphi)\setminus\setG}\neg \G\psi \bigg) \; \models_p \; \aft(\varphi, w_{0i})$$
holds for almost every $i \in \mathbb{N}$.
\end{reftheorem}
\begin{proof}
Follows from Lemma \ref {lem:master3} and Lemma \ref{lem:interpret}.
\qed\end{proof}

\begin{reftheorem}{thm:finalthm}
For any LTL formula $\varphi$, $\Lang(\A(\varphi))=\Lang(\varphi)$. 
\end{reftheorem}
\begin{proof}
($\Leftarrow$): If $w \models \varphi$, then let $\setG \subseteq \gsf(\varphi)$ be the set
of $\G$-subformulae $\G\psi$ of $\varphi$ such that $w \models \F\G\psi$, and let
$\pi$ be the mapping that assigns to every $\psi$ the rank $\pi(\psi)$ at which $w$ is accepted.
Further, let $\mathcal F(r_\psi,v)$ be the conjunction of the states of $\K(\psi)$ to which
$r_{\psi}$ assigns rank $\pi(\psi)$ or higher after reading $v$. 
By Theorem \ref{thm:master}
$$\bigg(\bigwedge_{\G\psi \in \setG} \big( \G\psi \; \wedge \; \mathcal F(r_\psi,w_{0i}) \big) 
\; \wedge \bigwedge_{\G\psi\in\gsf(\varphi)\setminus\setG}\neg \G\psi \bigg)
\models_p \aft(\varphi, w_{0i}) \ .$$
\noindent for almost every $i$. Since $\aft(\varphi, w_{0i})$ is the state reached by $\M(\varphi)$ after reading $w_{0i}$,
$\A(\varphi)$ accepts.

($\Rightarrow$): If $\A(\varphi)$ accepts, then it does so for a particular set $\setG \subseteq \gsf(\varphi)$
and ranking $\pi$. By Theorem \ref{thm:decomp2}, we have $w \models \F\G\psi$ for every $\G\psi \in \setG$.
%and $w \models \F\G \mathcal F(r_\psi)$. 
By the definition of the accepting condition,
$$\bigwedge_{\G\psi \in \setG} \big( \G\psi \; \wedge \; \mathcal F(r_\psi,w_{0i}) \big) \; \wedge \; \bigwedge_{\G\psi\in\gsf(\varphi)\setminus\setG}\neg \G\psi 
\models_p \aft(\varphi, w_{0i})$$
\noindent holds for almost every $i$. By Theorem \ref{thm:master}, $w \models \varphi$.
\end{proof}

% 
% \todo{Optimizations}
% \todo{irrelevance of $\G b$ in $\F (a \wedge \G b)$ - after $a$ start tracking in whatever initial state}
% \todo{acceptance condition: as in Rabinizer + compute as tree, not subsets}

\section{Further Experiments}\label{app:exp}

All automata were constructed within times of order of seconds, except for ltl2dstar where automata with thousands of states took several minutes and automata with hundreds of thousands of states took several hours. The timeout was set to one day except for the conjunction of 3 fairness constraints where it took a bit more than a day. Timeouts are denoted by ? and not applicability of the tool by $-$.

The first set of formulae is from the LTL$(\F,\G)$ fragment. The upper part comes from EEM (BEnchmarks for Explicit Model checkers)\cite{beem}, the lower from \cite{DBLP:conf/cav/SomenziB00} on which ltl2dstar was originally tested~\cite{DBLP:journals/tcs/KleinB06}. There are overlaps between the two sets. Note that the formula $(\F\F a\wedge \G\neg a)\vee(\G\G\neg a\wedge\F a)$ is a contradiction. Our method usually achieves the same results as the optimized LTL3DRA outperforming the first two approaches. 

{\footnotesize
$$\begin{array}{lrrrr}
 \text{Formula }	&	 \text{ltl2dstar}	&	\text{R.1}	&	\text{LTL3DRA}	&	\text{R.3}					\\	
 &\text{DRA}&\text{DRA}&\text{tGDRA}&\text{tGDRA}\\\hline
\G(a\vee\F b) 	&	4	&	4	&	2	&	2					\\	
\F\G a\vee \F\G b\vee \G\F c 	&	8	&	8	&	1	&	1					\\	
\F(a\vee b) 	&	2	&	2	&	2	&	2					\\	
\G\F(a\vee b) 	&	2	&	2	&	1	&	1					\\	
\G(a\vee b\vee c) 	&	3	&	2	&	2	&	2					\\	
\G(a\vee\F(b\vee c)) 	&	4	&	4	&	2	&	2					\\	
\F a\vee\G b 	&	4	&	3	&	3	&	3					\\	
\G(a\vee\F(b\wedge c)) 	&	4	&	4	&	2	&	2					\\	
(\F\G a\vee \G\F b) 	&	4	&	4	&	1	&	1					\\	
\G\F(a\vee b)\wedge \G\F(b\vee c) 	&	7	&	3	&	1	&	1					\\	\hline
(\F\F a\wedge \G\neg a)\vee(\G\G\neg a\wedge\F a) 	&	1	&	0	&	1	&	2					\\	
(\G\F a)\wedge\F\G b 	&	3	&	3	&	1	&	1					\\	
(\G\F a\wedge\F\G b)\vee(\F\G\neg a\wedge \G\F\neg b) 	&	5	&	4	&	1	&	1					\\	
\F\G a\wedge\G\F a 	&	2	&	2	&	1	&	1					\\	
\G(\F a \wedge \F b) 	&	5	&	3	&	1	&	3					\\	
\F a \wedge \F \neg a 	&	4	&	4	&	4	&	4					\\	
(\G(b\vee\G\F a)\wedge\G(c\vee\G\F\neg a))\vee\G b\vee\G c 	&	13	&	18	&	4	&	4					\\	
(\G(b\vee\F\G a)\wedge\G(c\vee\F\G\neg a))\vee\G b\vee\G c 	&	14	&	6	&	4	&	4					\\	
(\F(b\wedge\F\G a)\vee\F(c\wedge\F\G\neg a))\wedge\F b\wedge\F c 	&	7	&	5	&	4	&	4					\\	
(\F(b\wedge\G\F a)\vee\F(c\wedge\G\F\neg a))\wedge\F b\wedge\F c 	&	7	&	5	&	4	&	4	
\end{array}$$
}

The next set of LTL$(\F,\G)$ formulae are formulae whose satisfaction does not depend on any finite prefix of the word. They describe only ``infinitary'' behaviour. In this case, the master automaton has only one state. While DRA need to remember the last letter read, the transition-based acceptance together with the generalized acceptance condition allow tGDRA not to remember anything. Hence the number of states is $1$. The first two parts were used in \cite{cav12,DBLP:conf/atva/BabiakBKS13} and the third part in \cite{DBLP:conf/atva/BabiakBKS13}. The first part focuses on properties with fairness-like constraints. 

{\footnotesize
$$\begin{array}{lrrrr}
 \text{Formula }	&	 \text{ltl2dstar}	&	\text{R.1}	&	\text{LTL3DRA}	&	\text{R.3}					\\	
 &\text{DRA}&\text{DRA}&\text{tGDRA}&\text{tGDRA}\\\hline
(\F\G a\vee\G\F b) 	&	4	&	4	&	1	&	1					\\	
(\F\G a\vee\G\F b)\wedge(\F\G c\vee\G\F d) 	&	11\,324	&	18	&	1	&	1					\\	
\bigwedge_{i=1}^3(\G\F a_i\rightarrow\G\F b_i)	&	1\,304\,706	&	462	&	1	&	1					\\	
%(\bigwedge_{i=1}^5\G\F a_i)\rightarrow\G\F b 	&	?	&	64	&	1	&	1					\\	\hline
\G\F(\F a\vee\G\F b\vee\F\G(a\vee b)) 	&	14	&	4	&	1	&	1					\\	
\F\G(\F a\vee\G\F b\vee\F\G(a\vee b)) 	&	145	&	4	&	1	&	1					\\	
\F\G(\F a\vee\G\F b\vee\F\G(a\vee b)\vee\F\G b) 	&	181	&	4	&	1	&	1					\\	\hline
(\G \F a \vee \F \G b) 	&	4	&	4	&	1	&	1					\\	
(\G \F a \vee \F \G b) \wedge (\G \F b \vee \F \G c)	&	572	&	11	&	1	&	1					\\	
(\G \F a \vee \F \G b) \wedge (\G \F b \vee \F \G c) \wedge (\G \F c \vee \F \G d)	&	290\,046	&	52	&	1	&	1					\\	
(\G \F a \vee \F \G b) \wedge (\G \F b \vee \F \G c) \wedge (\G \F c \vee \F \G d) \wedge (\G \F d \vee \F \G h)	&	?	&	1288	&	1	&	1					\\	
\end{array}$$
}

In the next table, we have formulae of LTL$_{\setminus\G\U}$ used in \cite{atva13}. The first part comes mostly from the same sources and \cite{DBLP:conf/concur/EtessamiH00}. The second part is considered in \cite{atva13} in order to demonstrate the difficulties of the standard approach to handle
\begin{enumerate}
 \item many $\X$ operators inside the scope of other temporal operators, especially $\U$, where the slaves are already quite complex, and
 \item conjunctions of liveness properties where the efficiency of generalized Rabin acceptance condition may be fully exploited.
\end{enumerate}

{\footnotesize
$$\begin{array}{lrrrr}
 \text{Formula }	&	 \text{ltl2dstar}	&	\text{R.2}	&	\text{LTL3DRA}	&	\text{R.3}					\\	
 &\text{DRA}&\text{DRA}&\text{tGDRA}&\text{tGDRA}\\\hline
(\F p) \U (\G q)	&	4	&	3	&	2	&	2					\\	
(\G p) \U q	&	5	&	5	&	5	&	5					\\	
(p \vee q) \U p \vee \G q	&	4	&	3	&	3	&	3					\\	
\G(!p \vee \F q) \wedge ((\X p)\U q \vee \X((!p \vee !q)\U !p \vee \G(!p \vee !q))) 	&	19	&	8	&	-	&	5					\\	
\G(q \vee \X \G p) \wedge \G(r \vee \X \G !p)	&	5	&	14	&	4	&	5					\\	
(\X(\G r \vee r \U (r \wedge s \U p)))\U(\G r \vee r \U (r\wedge s))	&	18	&	9	&	8	&	8					\\	
p \U(q\wedge \X(r\wedge (\F(s\wedge \X(\F(t\wedge \X(\F(u\wedge \X \F v))))))))	&	9	&	13	&	13	&	13					\\	\hline
(\G \F (a \wedge \X \X b) \vee \F \G b) \wedge \F \G (c \vee (\X a \wedge \X \X b))	&	353	&	73	&	-	&	12					\\	
\G \F (\X \X \X a \wedge \X \X \X \X b) \wedge \G \F (b \vee \X c) \wedge \G \F (c \wedge \X \X a)	&	2\,127	&	169	&	-	&	16					\\	
(\G \F a \vee \F \G b) \wedge (\G \F c \vee \F \G (d \vee \X e))	&	18\,176	&	80	&	-	&	2					\\	
(\G \F (a \wedge \X \X c) \vee \F \G b) \wedge (\G \F c \vee \F \G (d \vee \X a \wedge \X \X b))	&	?	&	142	&	-	&	12					\\	
a \U b \wedge (\G \F a \vee \F \G b) \wedge (\G \F c \vee \F \G d) \vee 	&	640\,771	&	210	&	8	&	7	\\
\qquad\qquad\vee a \U c \wedge (\G \F a \vee \F \G d) \wedge (\G \F c \vee \F \G b)
\end{array}$$
}

The following randomly picked two examples illustrate the same two phenomena as in the previous table now on general LTL formulae.

{\footnotesize
$$\begin{array}{lrrrr}
 \text{Formula }	&	 \text{ltl2dstar}	&	\text{R.1/2}	&	\text{LTL3DRA}	&	\text{R.3}					\\	
 &\text{DRA}&\text{DRA}&\text{tGDRA}&\text{tGDRA}\\\hline
\F \G ((a \wedge \X \X b \wedge \G \F b)\U(\G(\X \X ! c \vee \X \X (a\wedge b))))	&	2\,053	&	-	&	-	&	11					\\	
\G(\F !a \wedge \F (b \wedge \X !c ) \wedge \G \F (a \U d)) \wedge \G \F ((\X d)\U(b \vee \G c))	&	283	&	-	&	-	&	7					\\	
\end{array}$$
}

The last set contains %one general LTL formula from each of \cite{DBLP:conf/cav/SomenziB00} and \cite{DBLP:conf/concur/EtessamiH00}. The next two formulae are examples 
two examples of formulae from a network monitoring project Liberouter (https://www.liberouter.org/).
The subsequent 5 more complex formulae are from \textsc{Spec Pattern} \cite{DBLP:conf/icse/DwyerAC99} (available at \cite{specpatterns}) and express the following ``after Q until R'' properties:
\begin{enumerate}
 \item[$\varphi_{35}:$] $\G (!q \vee (\G p \vee (!p \U (r \vee (s \wedge !p \wedge \X(!p \U t))))))$
 \item[$\varphi_{40}:$] $\G(!q \vee (((!s \vee r) \vee \X(\G(!t\vee r) \vee !r \U (r\wedge (!t \vee r)))) \U (r \vee p) \vee \G((!s \vee \X \G !t))))$
 \item[$\varphi_{45}:$] $\G (!q \vee (!s \vee \X(\G !t \vee !r \U (r\wedge !t)) \vee \X(!r \U (r \wedge \F p))) \U (r \vee \G (!s \vee \X(\G !t \vee !r \U (r\wedge !t)) \vee \X(!r \U (t \wedge \F p)))))$
 \item[$\varphi_{50}:$] $\G (!q \vee (!p \vee (!r \U (s \wedge !r \wedge \X(!r \U t)))) \U (r \vee \G (!p \vee (s \wedge \X \F t))))$
 \item[$\varphi_{55}:$] $\G (!q \vee (!p \vee (!r \U (s \wedge !r \wedge !z \wedge \X((!r \wedge !z) \U t)))) \U (r \vee \G (!p \vee (s \wedge !z \wedge \X(!z \U t)))))$
\end{enumerate}

{\footnotesize
$$\begin{array}{lrrrr}
 \text{Formula }	&	 \text{ltl2dstar}	&	\text{R.1/2}	&	\text{LTL3DRA}	&	\text{R.3}					\\	
 &\text{DRA}&\text{DRA}&\text{tGDRA}&\text{tGDRA}\\\hline
\G (((!p1))\wedge (p2 \U ((!p2) \U ((!p3) \vee p4))))	&	7	&	-	&	-	&	4					\\
 \G (((p1) \wedge \X ! p1) \vee \X (p1 \U (((!p2) \wedge p1) \wedge 	&	8	&	-	&	-	&	8					\\
\quad\quad \X (p2 \wedge p1 \wedge (p1 \U (((!p2) \wedge p1) \wedge \X (p2 \wedge p1))))))) \\\hline
% %\G (!q \vee (\G p \vee (!p \U (r \vee (s \wedge !p \wedge \X(!p \U t))))))	&		&		&		&						\\	
% pc 1.5 	35&	6	&	-	&	-	&	6					\\	
% %\G(!q \vee (((!s \vee r) \vee \X(\G(!t\vee r) \vee !r \U (r\wedge (!t \vee r)))) \U (r \vee p) \vee \G((!s \vee \X \G !t))))	&		&		&		&						\\	
% pc 2.5 	40&	314	&	-	&	-	&	32					\\	
% %\G (!q \vee (!s \vee \X(\G !t \vee !r \U (r\wedge !t)) \vee \X(!r \U (r \wedge \F p))) \U (r \vee \G (!s \vee \X(\G !t \vee !r \U (r\wedge !t)) \vee \X(!r \U (t \wedge \F p)))))	&		&		&		&						\\	
% rc 1.5	45&	1450	&	-	&	-	&	78					\\	
% %\G (!q \vee (!p \vee (!r \U (s \wedge !r \wedge \X(!r \U t)))) \U (r \vee \G (!p \vee (s \wedge \X \F t))))	&		&		&		&						\\	
% rc2.5	50&	28	&	-	&	-	&	23					\\	
% %\G (!q \vee (!p \vee (!r \U (s \wedge !r \wedge !z \wedge \X((!r \wedge !z) \U t)))) \U (r \vee \G (!p \vee (s \wedge !z \wedge \X(!z \U t)))))	&		&		&		&						\\	
% ccp	55&	28	&	-	&	-	&	23					\\
\varphi_{35}:\text{2 cause-1 effect precedence chain}&	6	&	-	&	-	&	6					\\	
\varphi_{40}:\text{1 cause-2 effect precedence chain}&	314	&	-	&	-	&	32					\\	
\varphi_{45}:\text{2 stimulus-1 response chain}&	1450	&	-	&	-	&	78					\\	
\varphi_{50}:\text{1 stimulus-2 response chain}&	28	&	-	&	-	&	23						
\\\varphi_{55}:\text{1-2 response chain constrained by a single proposition}&	28	&	-	&	-	&	23		
\end{array}$$
}

\end{document}